\documentclass[twoside,final,11pt]{entics}
\usepackage{enticsmacro}
\usepackage{graphicx}
\usepackage[all]{xy}
\usepackage{macros}

\newtheorem{ques}[thm]{Question}

\usepackage[capitalize]{cleveref}
\Crefname{remark}{Remark}{Remarks}
\crefname{remark}{remark}{remarks}
\Crefname{question}{Question}{Questions}
\crefname{question}{question}{questions}

\def\conf{MFPS 2026} 	
\volume{6}			
\def\volu{6}			
\def\papno{3}			
\def\mydoi{19622}
\def\lastname{Cavallo, Höfer} 
\def\runtitle{Univalence without finction extensionality} 
\def\copynames{E. Cavallo, J. Höfer}    
\begin{document}
\begin{frontmatter}
 \title{Univalence Without Function Extensionality\thanksref{ALL}} 
 \thanks[ALL]{
   We thank Lorenzo Perticone for first (inadvertently) calling our attention to this question, and we thank the Gothenburg Logic and Types unit for many lunchtime discussions on the topic.
   We also thank Andr{\'a}s Kov{\'a}cs for his Agda formalization of the polynomial model, which was of great help to us in understanding the construction.}   
  \author{Evan Cavallo\thanksref{b}\thanksref{myemail}\thanksref{KAW}} 
  \author{Jonas Höfer\thanksref{b}\thanksref{coemail}\thanksref{ERC}} 
  \thanks[myemail]{Email: \href{mailto:evan.cavallo@gu.se} {\texttt{\normalshape evan.cavallo@gu.se}}}
  \address[b]{Department of Computer Science and Engineering\\University of Gothenburg and Chalmers University of Technology
    \\ Gothenburg, Sweden} 
  \thanks[coemail]{Email: \href{mailto:hoferj@chalmers.se} {\texttt{\normalshape hoferj@chalmers.se}}}
  \thanks[KAW]{Supported by the Knut and Alice Wallenberg Foundation (KAW), Grant No.\ 2019.0116}
  \thanks[ERC]{Supported by ForCUTT project, ERC advanced grant 10105329}
\begin{abstract}
  It is a well-known theorem of homotopy type theory, originally due to Voevodsky, that function extensionality holds inside any univalent universe.
  We consider a weaker variant of the univalence axiom, asserting that the wild category formed by the universe is univalent, which we call categorical univalence.
  We show that categorical univalence does not imply function extensionality by an analysis of Von Glehn's polynomial model construction, which produces models of Martin-L\"of type theory that always refute function extensionality.
  We find in particular that when the base model has a univalent universe, its polynomial model has a universe that is categorically univalent but lacks function extensionality.
\end{abstract}
\end{frontmatter}

\section{Introduction}\label{sec:introduction}

In 2010, Voevodsky~\cite{Voevodsky2010Funext,Voevodsky2010NSF} discovered that any universe of intensional Martin-L\"of type theory (\(\ITT\)) satisfying his univalence axiom also satisfies function extensionality: (dependent) functions between types in the universe are equal as soon as they are homotopic.
This result became a foundational pillar of Homotopy Type Theory / Univalent Foundations.
For constructivists, it was an additional motivation to justify univalence constructively---noted for example by Bezem, Coquand, and Huber \cite{BezemCoquandHuber2013}---given the historical difficulty of integrating function extensionality with constructive type theory.

At the same time, the connection between univalence and function extensionality has always seemed contingent.
It is unclear whether univalence implies extensionality principles for other negative type formers, such as coinductive \cite{Voevodsky2014Coinductives} or modal \cite[Conjecture 11.2.2]{Gratzer2023} types, which suggests functions might be privileged simply because they appear in the statement of univalence.
Furthermore, minor variations on the univalence axiom are not known to imply function extensionality.

In a post on MathOverflow in 2013 \cite{Dorais2013}, Fran\c{c}ois G.\ Dorais proposed\footnote{With some input from Mike Shulman.} one such variation.
To contextualize Dorais' axiom, let us first review the standard definitions.
For functions \(f,g \co \prod_{a \co A} B(a)\), we write \(f \sim g \coloneqq \prod_{a \co A} fa =_{B(a)} ga\) for the type of \emph{homotopies} from \(f\) to \(g\).
The type \(A \simeq B\) of \emph{(homotopy) equivalences} between types \(A,B\) is the type of homotopy bi-invertible maps, that is, \(f \co A \to B\) equipped with \(s,r \co B \to A\) such that \(fs \sim \id_B\) and \(rf \sim \id_A\) \cite[\S9.2]{Rijke2025}.
We assume a fixed universe \(\U\).

\begin{definition}\label[definition]{defn:function-extensionality}
  \emph{Function extensionality} (\(\FE\)) is the principle that for every family of types \(a \co A \vdash B(a)\) and \(f,g \co \prod_{a \co A} B(a)\), the map \((f =_{\prod_{a \co A} B(a)} g) \to (f \sim g)\) is an equivalence.
  We write \(\FE_\U\) for the relativization of \(\FE\) to \(\U\), i.e., its restriction to the case where \(A \co \U\) and \(B \co A \to \U\).
\end{definition}

\begin{definition}\label[definition]{defn:univalence}
  \emph{Univalence} (\(\UA_\U\)) is the principle that the map \(\idToEquiv \co (A =_\U B) \to (A \simeq B)\) is an equivalence for all \(A,B \co \U\).
\end{definition}

Dorais observed essentially that the map \(\idToEquiv \co (A =_\U B) \to (A \simeq B)\), which sends the reflexive path to the identity equivalence, factors up to homotopy through an intermediate type
\[
  \begin{tikzcd}[row sep=0em]
    & (A \catIso B) \ar[dashed,near start]{dr}{\CEqToEq} \\
    (A =_\U B) \ar[dashed,near end]{ur}{\idToCEq} \ar{rr}[below]{\idToEquiv} && (A \simeq B)
  \end{tikzcd}
\]
of what we call \emph{categorical equivalences}: maps \(f \co A \to B\) equipped with \(s,r \co B \to A\) such that \(fs =_{B\to B} \id_B\) and \(rf =_{A\to A} \id_A\), i.e., with left and right inverses up to equality rather than homotopy.
This suggests Dorais' proposed weakening of univalence:

\begin{definition}\label[definition]{defn:categorical-univalence}
  \emph{Categorical univalence} (\(\CUA_\U\)) is the principle that \(\idToCEq \co (A =_{\U} B) \to (A \catIso B)\) is an equivalence for all \(A,B \co \U\).
\end{definition}

The type \(A \CEq B\) can be described as the type of isomorphisms from \(A\) to \(B\) in the \emph{wild category} of types in \(\U\) and functions between them.
\(\CUA_\U\) states exactly that \(\U\) is a univalent wild category.
In the presence of function extensionality in \(\U\), the map \((A \catIso B) \to (A \simeq B)\) is an equivalence, and so \(\FE_{\U} + \CUA_{\U}\) implies \(\UA_{\U}\); conversely, the fact that \(\UA_{\U}\) implies \(\FE_{\U}\) means that it also implies \(\CUA_{\U}\).
Dorais asked whether the converse is true: does \(\CUA_{\U}\) imply \(\UA_{\U}\), or equivalently \(\FE_{\U}\)?

We answer this question in the negative, identifying a model of \(\ITT\) with a universe that validates \(\CUA_{\U}\) but not \(\FE_{\U}\).
Actually, we prove the consistency of \(\lnot\FE_{\U}\) with a slightly stronger statement:

\begin{definition}\label[definition]{defn:familial-categorical-univalence}
  \emph{Familial categorical univalence} (\(\FamCUA_\U\)) is the principle that for all \(I \co \U\), the wild category \(\U^I\)---whose objects are families \(A \co I \to \U\) and whose morphisms \(A \to B\) are families of functions \(\prod_{i\co I} A(i) \to B(i)\)---is a univalent wild category.
\end{definition}

We assume strict \(\eta\) laws for unit and \(\Pi\) types, so \(\FamCUA_\U\) implies \(\CUA_\U\) by taking \(I = 1\).
We show the independence of \(\FamCUA_{\U}\) from \(\FE_{\U}\) using Von Glehn's \emph{polynomial model} construction \(\Poly{-}\) \cite{Glehn2015,MossGlehn2018}, a known source of models of type theory that refute function extensionality.
Specifically, we prove:

\begin{theorem}[{\ref{prop:poly-inherits-famcua}}]
  Let \(\cat{C}\) be a model of \(\ITT\) with extensive finite coproducts of types satisfying the strict \(\eta\) rule.
  If \(\cat{C} \vDash \FamCUA_{\U}\), then \(\Poly{\cat{C}} \vDash \FamCUA_{\U}\).
\end{theorem}

Familial categorical univalence arises naturally in the construction: just to show \(\Poly{\cat{C}} \vDash \CUA_{\U}\), we already require \(\cat{C} \vDash \FamCUA_{\U}\).
Function extensionality always fails in polynomial models \cite[\S4.5]{Glehn2015}, so it remains to provide a suitable input model.
Off-the-shelf cubical and simplicial models of homotopy type theory will do, as Moss and Von Glehn have already observed \cite[\S6]{MossGlehn2018}.
We conclude:

\begin{theorem}[{\ref{prop:famcua-not-implies-fe}}]
  \(\ITT + \FamCUA_{\U} \nvdash \FE_{\U}\).
\end{theorem}

Part of the appeal of weak foundations is that they allow us to tease apart the components of mathematics.
Each type former of Martin-L\"of's type theory has a distinct, well-defined purpose.
Univalence fits in uneasily in this picture.
While it has beautiful consequences, it also has \emph{many} consequences, and---like impredicativity or the law of the excluded middle---it may be hiding finer structure beneath its surface.

By scratching at that surface, we hope to understand what makes univalence tick.
The polynomial model offers some motivation and a testing ground for weaker forms of the axiom.
We are left with more questions than answers; unlike in the case with \(\FE\), where superficial variations on univalence usually turn out to be equivalent, here we find subtly distinct axioms with no canonical choice among them.
Still, we hope our results can provoke further reflection on the foundations of homotopy type theory.

\subsection{Outline}

In \Cref{sec:decomposing-univalence}, we recall some basic definitions, then observe that \(\FE_{\U}\) holds if and only if the canonical map \(\CEqToEq \co (A \CEq B) \to (A \simeq B)\) is an equivalence for all \(A,B \co \U\), meaning that univalence quite literally factors into function extensionality and categorical univalence.
In \Cref{sec:von-glehn-polynomial} we recall Von Glehn's polynomial model construction.
The main technical contribution is \Cref{sec:categorical-univalence}, where we show that the polynomial model \(\Poly{\cat{C}}\) inherits \(\FamCUA_{\U}\) from the base model \(\cat{C}\).
In \Cref{sec:polynomial-refutation} we apply this result to a univalent base model to conclude that \(\ITT + \FamCUA_{\U} \nvdash \FE_{\U}\).
We discuss and compare other possible weakenings of \(\UA_{\U}\) in \Cref{sec:other-univalences}, and finish with a review of related work in \Cref{sec:related-work}.

\section{Decomposing univalence}\label{sec:decomposing-univalence}

Our basic theory \(\ITT\) is Martin-L\"of type theory with \(\Sigma\) types, \(\Pi\) types, intensional identity types, binary coproduct types, empty and unit types, and one universe \(\U\) closed under all of these type formers.\footnote{There is no issue extending our results to multiple universes, but we only need one.}
We use the term \emph{strict equality} and symbol \(\seq\) for equality on the judgmental level; we use \(=\) for identity types.
We use \(\sequiv\) for strict isomorphisms: two functions in opposite directions composing strictly to the identity.
Note that \(A \sequiv B\) is not a type, and \(e \co A \sequiv B\) is merely a shorthand for a meta-level assumption.
Besides strict \(\beta\) rules for all type formers, we include strict \(\eta\) rules for \(\Sigma\) types, \(\Pi\) types, and the unit type.

For basic results, we cite Rijke's book~\cite{Rijke2025}, which does not introduce \(\FE\) until Chapter 13; we only use results from earlier chapters.
Crucially, we have basic facts about contractible types and that \(\Sigma\) types respect equivalences in both arguments.
Note that the analogous fact does not hold for \(\Pi\) types absent \(\FE\).
In contrast to Rijke \cite{Rijke2025}, we assume the strict \(\eta\) rule for \(\Sigma\) types, not only for \(\Pi\) types.
This means, for example, that the equivalence witnessing the distributivity of \(\Pi\) types over \(\Sigma\) types is a strict isomorphism.

\subsection{Univalent wild categories}

The universe of \(\ITT\) has the structure of an \((\infty,1)\)-category, with type of objects \(\U_0 \coloneqq \U\) and type of morphisms \(\U_1(A,B) \coloneqq (A \to B)\).
The first layer of such an \((\infty,1)\)-categorical structure is captured by the Capriotti and Kraus' notion of \emph{wild category} \cite[Definition 4.1]{CapriottiKraus2017}.

\begin{definition}
  A \emph{wild category}\footnote{Capriotti and Kraus call this a \emph{wild precategory}.} \(\cat{C}\) is a type \(\cat{C}_0\) and family of types \(x,y \co \cat{C}_0 \vdash \cat{C}_1(x,y)\) equipped with
  \begin{enumerate}
  \item composites \(g \cc f \co \cat{C}_1(x,z)\) for all \(g \co \cat{C}_1(y,z)\), \(f \co \cat{C}_1(x,y)\),
  \item identities \(\id_{x} \co \cat{C}_1(x,x)\) for all \(x \co \cat{C}_0\),
  \item associators \(\alpha_{h,g,f} \co h \cc (g \cc f) = (h \cc g) \cc f\) for all \(h \co \cat{C}_1(z,w)\), \(g \co \cat{C}_1(y,z)\), \(f \co \cat{C}_1(x,y)\), and
  \item unitors \(\lambda_f \co \id_y \cc f = f\) and \(\rho_f \co f \cc \id_x = f\) for all \(f \co \cat{C}_1(x,y)\).
  \end{enumerate}
  If clear from context, we omit the subscripts when referring to the type of objects or family of morphisms.
  We write \(x \to y\) for \(\cat{C}_1(x,y)\) when \(\cat{C}\) is clear, and we sometimes write \(gf\) for \(g \cc f\).
\end{definition}

\begin{example}\label[example]{ex:our-wild-categories}
  As noted above, the universe \(\U\) has a wild category structure with \(\U(A,B) \coloneqq (A \to B)\), composition and identities given by the usual composition of functions and identity functions, and reflexive equalities for the associators and unitors.
  More generally, for every type \(I\) there is a wild category \(\U^I\) whose objects are families \(A \co I \to \U\) and whose morphisms are indexed functions, \(\U^I(A,B) \coloneqq \prod_{i\co I} A(i) \to B(i)\).

  These wild categories are really strictly coherent \((\infty,1)\)-categories: the associators and unitors are strict equalities and satisfy all higher coherence laws (e.g., the pentagon) up to strict equality.
  All of the concrete wild categories we encounter in this article are of this kind.
\end{example}

Importantly, wild-categorical structure suffices to define \emph{isomorphism}.

\begin{definition}
  Given \(s \co x \to y\) and \(r \co y \to x\) in a wild category \(\cat{C}\), we say that \(r\) is a \emph{retraction of \(s\)} and \(s\) is a \emph{section of \(r\)} if \(rs = \id_x\).
  For a morphism \(f\), we write \(\Sec(f)\) and \(\Ret(f)\) for the types of sections and retractions of \(f\) respectively.
  We say \(f\) is a \emph{\(\cat{C}\)-isomorphism} if we have an element of the type \(\isCatIso[\cat{C}](f) \coloneqq \Sec(f) \times \Ret(f)\) and write \(x \catIso[\cat{C}] y\) for the type of isomorphisms between two objects \(x, y \co \cat{C}\).
\end{definition}

\begin{remark}
  The isomorphisms in the wild category \(\U\) are exactly the categorical equivalences introduced in \Cref{sec:introduction}.
  To avoid confusion, we refer to a pair of functions \(s \co A \to B\) and \(r \co B \to A\) between types such that \(r s \sim \id_A\) as a \emph{homotopy section} and \emph{homotopy retraction} respectively.
  The term \emph{isomorphism} is sometimes used in the literature to refer to what the HoTT Book~\cite{HoTTBook2013} calls \emph{quasi-inverses}, that is, maps \(f \co A \to B\) and \(g \co B \to A\) with homotopies \(gf \sim \id_A\) and \(fg \sim \id_B\).
  We never use the term in this way.
\end{remark}

\begin{remark}
  Categorical equivalences are related to the point-free propositions of Donk\'{o} and Kaposi~\cite{DonkoKaposi2022}.
  A type \(A\) is a point-free proposition if the first and second projection \(A \times A \to A\) are equal as functions.
  This is equivalent to one of the projections \(A \times A \to A\) or the diagonal \(A \to A \times A\) being a categorical equivalence.
  Their dependent point-free propositions can be similarly characterized in \(\U^I\).
\end{remark}

The following holds by a Yoneda style argument.

\begin{lemma}\label[lemma]{prop:cat-isomorphism-characterization}
  For a morphism \(f \co x \to y\) in a wild category \(\cat{C}\), the following are logically equivalent:
  \begin{enumerate}
    \item \(f\) is an isomorphism,
    \item \(f^* \co \cat{C}(x, z) \to \cat{C}(y, z)\) is an equivalence for all \(z \co \cat{C}\),
    \item \(f_* \co \cat{C}(z, y) \to \cat{C}(z, x)\) is an equivalence for all \(z \co \cat{C}\).
  \end{enumerate}
\end{lemma}

\begin{lemma}\label[lemma]{prop:sections-of-equivalence-are-unqiue}
  Given an isomorphism \(f \co x \to y\) in a wild category, the type \(\Sec(f)\) is contractible.
\end{lemma}
\begin{proof}
  Denote by \(f^{-1}\) the retraction of \(f\).
  By \Cref{prop:cat-isomorphism-characterization} and \cite[Exercise~9.1]{Rijke2025} we have the equivalences \( \paren{\sum_{g \co y  \to x} fg = \id_y} \simeq \paren{\sum_{g \co y \to x} f^{-1}(fg) = f^{-1}\id} \simeq \paren{\sum_{g \co y \to x} g = f^{-1}} \).
  In the last step we use that composition with a path is an equivalence.
  The last type is contractible.
\end{proof}
\begin{corollary}\label[corollary]{prop:being-isomorphism-is-proposition}
  For a morphism \(f\) in a wild category, \(\isCatIso(f)\) is a proposition.
\end{corollary}
\begin{proof}
  We show that \(\isCatIso(f)\) is contractible, assuming that \(f\) is an isomorphism~\cite[Proposition 12.1.3]{Rijke2025}.
  By \Cref{prop:sections-of-equivalence-are-unqiue} and its dual, both \(\Sec(f)\) and \(\Ret(f)\) are contractible. 
\end{proof}

\begin{lemma}\label[lemma]{prop:isomorphisms-satisfy-2-out-of-3}
  Isomorphisms in a wild category satisfy \(2\)-out-of-\(3\).
\end{lemma}
\begin{proof}
  By associativity we have \((fg)^* \sim f^* g^*\).
  The structure of an equivalence transfers across homotopies.
  Hence, the claim follows from \(2\)-out-of-\(3\) for equivalences~\cite[Exercise~9.4]{Rijke2025}.
\end{proof}

For every object \(x \co \cat{C}\) in a wild category, the identity \(\id_x \co x \to x\) is an isomorphism.
By path induction, we may generalize to a map \(\idToCatIso \co x =_{\cat{C}} y \to x \catIso[\cat{C}] y\) for \(x,y \co \cat{C}\).
We define \cite[Definition 4.16]{CapriottiKraus2017}:

\begin{definition}
  A wild category \(\cat{C}\) is \emph{univalent} if \(\idToCatIso \co x =_{\cat{C}} y \to x \catIso[\cat{C}] y\) is an equivalence for \(x,y \co \cat{C}\).
\end{definition}
\begin{lemma}\label[lemma]{prop:contractability-characterization-of-univalence}
  A wild category \(\cat{C}\) is univalent exactly if \(\sum_{y \co \cat{C}} x \catIso[\cat{C}] y\) is contractible for all \(x\).
\end{lemma}
\begin{proof}
  By the fundamental theorem of identity types~\cite[Theorem 11.2.2]{Rijke2025}.
\end{proof}

Univalence of a universe \(\U\) (\(\UA_{\U}\), \cref{defn:univalence}) cannot be formulated on the level of an arbitrary wild category, as it refers to homotopy of functions.
As \(\UA_{\U}\) implies \(\FE_{\U}\), however, it also implies that \(\U\) is a univalent wild category.
Absent \(\FE_{\U}\), the converse may fail: as we will see, \(\U\) can be a univalent wild category without \(\idToEquiv\) being an equivalence.
We can consider ordinary univalence as the conjunction of two equivalences: one between \(A =_{\U} B\) and \(A \cong_{\U} B\) and one between \(A \cong_{\U} B\) and \(A \simeq B\).

\subsection{Categorical equivalences and function extensionality}
\label{sec:decomposing-univalence:funext}

In comparing \(A \cong_{\U} B\) and \(A \simeq B\), it is natural to forget about universes entirely.
Recall from \Cref{sec:introduction} that a function \(f \co A \to B\) between (possibly large) types is a \emph{categorical equivalence} if it admits a section and retraction, that is, \(s,r \co B \to A\) with \(fs = \id\) and \(rf = \id\).
We write \(\isCEq(f)\) for the type of witnesses that \(f\) is a categorical equivalence, and \(A \CEq B\) for the type of categorical equivalences from \(A\) to \(B\).

In \(\ITT\), by canonicity, the only closed categorical equivalences are the strict isomorphisms, such as \(A \times B \sequiv B \times A\).
With \(\CUA_{\U}\) (\Cref{defn:categorical-univalence}), there are more; for example, any \(e \co A \cong B\) in \(\U\) yields \((a =_A a') \cong (ea =_B ea')\) for \(a,a' \co A\).
The map \(\CEqToEq \co (A \catIso B) \to (A \simeq B)\), which converts equalities in function types to homotopies, becomes an equivalence under \(\FE\).
In fact, it is an equivalence \emph{only} if \(\FE\) holds.

\begin{definition}\label[definition]{defn:equivalence-improvement}
  \emph{Equivalence improvement} (\(\EI\)) is the principle that for all types \(A,B\), the map \(\CEqToEq \co (A \catIso B) \to (A \simeq B)\) is an equivalence.
\end{definition}

We recall a lemma familiar from proofs that \(\UA_{\U}\) implies \(\FE_{\U}\) \cite[Theorem 4.9.4]{HoTTBook2013} \cite[Theorem 17.3.2]{Rijke2025}:

\begin{lemma}\label[lemma]{prop:dependent-product-is-equivalent-to-fiber}
  For any family of types \(a \co A \vdash B(a)\), the type \(\prod_{a \co A} B(a)\) is equivalent to the fiber of \(\pi_* \co \left(\sum_{a \co A} B(a)\right)^A \to A^A\) at \(\id_A \co A^A\), where \(\pi_*\) is post-composition with the first projection.
  Equivalently, the following strictly commutative square is a homotopy pullback:
  \[
    \begin{tikzcd}
      {\displaystyle\prod_{a \co A} B(a)}
      \ar[d]
      \ar[r, "\pair{\id_A,-}"]
      \ar[phantom,start anchor=center,to path={-- ++({0-45}:8ex) \tikztonodes}, "\lrcorner"]
    &
      {\displaystyle\paren[\bigg]{\sum_{a \co A} B(a)}^A}
      \ar[d, "\pi_*"]
    \\
      {1}
      \ar[r, "\id_A"']
    &
      {A^A\rlap{.}}
    \end{tikzcd}
  \]
\end{lemma}
\begin{proof}
  We prove that the fiber of \(\pi_*\) at an arbitrary \(t \co A \to A\) is equivalent to \(\prod_{a \co A} B(ta)\).
  Define
  \begin{mathpar}
    \begin{array}{l}
      s \co \left(\prod_{a \co A} B(ta)\right) \to \fiber{\pi_*}(t) \\
      f \mapsto \pair{\lambda a.\pair{ta,fa},\refl}
    \end{array}
    \and
    \begin{array}{l}
      r \co \fiber{\pi_*}(t) \to \prod_{a \co A} B(ta) \\
      \pair{g,\refl} \mapsto \pi_1 \circ g
    \end{array}
  \end{mathpar}
  where \(\pi_1\) is the second projection from \(\sum_{a\co A} B(a)\).
  We have \(rs \seq \id\) and a homotopy \(sr \sim \id\) by path induction.
\end{proof}

\begin{theorem}\label[theorem]{prop:funext-characterizations}
  In \(\ITT\), the following are logically equivalent:
  \begin{enumerate}
  \item\label{prop:funext-characterizations:equiv} \(\EI\)\emph{:} for all types \(A,B\), \(\CEqToEq \co (A \catIso B) \to (A \simeq B)\) is an equivalence,
  \item\label{prop:funext-characterizations:section} for all types \(A,B\), \(\CEqToEq \co (A \catIso B) \to (A \simeq B)\) admits a homotopy section,
  \item\label{prop:funext-characterizations:isequiv-is-prop} for all types \(A,B\) and every \(f \co A \to B\), the type \(\isEquiv(f)\) is a proposition,
  \item\label{prop:funext-characterizations:improve-homotopy-to-id} for every type \(A\) and \(f \co A \to A\), if \(f \sim \id_A\) then \(f = \id_A\),
  \item\label{prop:funext-characterizations:isequiv} for all types \(A,B\) and every \(f \co A \to B\), we have \(\isEquiv(f) \to \isCEq(f)\),
  \item\label{prop:funext-characterizations:comp} for all types \(A,B\) and every \(f \co A \to B\), we have \(\isEquiv(f) \to \isEquiv(f_*)\),
  \item\label{prop:funext-characterizations:weak} \emph{Weak \(\FE\):} for every family of contractible types \(a \co A \vdash P(a)\), the type \(\prod_{a \co A} P(a)\) is contractible,
  \item\label{prop:funext-characterizations:funext} \(\FE\)\emph{:} for every \(a \co A \vdash B(a)\) and \(f,g \co \prod_{a \co A} B(a)\), the map \((f = g) \to (f \sim g)\) is an equivalence.
  \end{enumerate}
\end{theorem}
\begin{proof}
  That~\eqref{prop:funext-characterizations:equiv}~\(\implies\)~\eqref{prop:funext-characterizations:section} is immediate.
  For~\eqref{prop:funext-characterizations:section}~\(\implies\)~\eqref{prop:funext-characterizations:isequiv-is-prop}, note that any homotopy section of \(\CEqToEq\) exhibits \(\isEquiv(f)\) as a homotopy retract of the proposition \(\isCEq(f)\) (cf.~\cref{prop:being-isomorphism-is-proposition}).
  For~\eqref{prop:funext-characterizations:isequiv-is-prop}~\(\implies\)~\eqref{prop:funext-characterizations:improve-homotopy-to-id}, observe that a homotopy \(f \sim \id_A\) implies that \(f\) is a homotopy inverse of \(\id_A\).
  As \(\id_A\) is also its own homotopy inverse, \eqref{prop:funext-characterizations:isequiv-is-prop} implies that \(f = \id_A\).
  That~\eqref{prop:funext-characterizations:improve-homotopy-to-id}~\(\implies\)~\eqref{prop:funext-characterizations:isequiv} is immediate, and \eqref{prop:funext-characterizations:isequiv}~\(\implies\)~\eqref{prop:funext-characterizations:comp} follows from \Cref{prop:cat-isomorphism-characterization}.
  The implication~\eqref{prop:funext-characterizations:comp}~\(\implies\)~\eqref{prop:funext-characterizations:weak},which appears in standard proofs of \(\FE_{\U}\) from \(\UA_{\U}\)~\cite[Theorem 4.9.4]{HoTTBook2013} \cite[Theorem 17.3.2]{Rijke2025}, follows from \Cref{prop:dependent-product-is-equivalent-to-fiber} and the fact that the fibers of an equivalence are contractible~\cite[Theorem 10.4.6]{Rijke2025}.
  That~\eqref{prop:funext-characterizations:weak}~\(\implies\)~\eqref{prop:funext-characterizations:funext} is due to Voevodsky; see for example \cite[Theorem 4.9.5]{HoTTBook2013} or \cite[Theorem 13.1.2]{Rijke2025}.
  That~\eqref{prop:funext-characterizations:funext}~\(\implies\)~\eqref{prop:funext-characterizations:equiv} is by definition of \(\CEqToEq\).
\end{proof}

\Cref{prop:funext-characterizations} relativizes to \(\U\).
From this we recover that \(\UA_{\U}\) implies \(\FE_{\U}\) (and thus \(\CUA_{\U}\)).

\begin{corollary}
  \(\ITT \vdash \UA_{\U} \leftrightarrow (\CUA_{\U} \land \FE_{\U})\).
\end{corollary}
\begin{proof}
  For \(A,B \co \U\), consider the following homotopy commutative triangle.
  \[
    \begin{tikzcd}[row sep=0em]
      & (A \CEq B) \ar[dashed,near start]{dr}{\CEqToEq} \\
      (A =_\U B) \ar[dashed,near end]{ur}{\idToCEq} \ar{rr}[below]{\idToEquiv} && (A \simeq B)
    \end{tikzcd}
  \]
  The map \(\CEqToEq\) is an equivalence for all \(A,B \co \U\) if and only if \(\FE_{\U}\) holds, by \eqref{prop:funext-characterizations:equiv} \(\iff\) \eqref{prop:funext-characterizations:funext} of \Cref{prop:funext-characterizations} in \(\U\).
  Thus, if both \(\CUA_{\U}\) and \(\FE_{\U}\) hold, then \(\UA_{\U}\) holds by 2-out-of-3 for equivalences.

  Conversely, if \(\UA_\U\) holds, then \(\idToEquiv\) has in particular a homotopy section for all \(A,B \co \U\).
  Post-composing these homotopy sections with \(\idToCEq\) yields homotopy sections of \(\CEqToEq\) for all \(A,B \co \U\).
  By \eqref{prop:funext-characterizations:section} \(\implies\) \eqref{prop:funext-characterizations:funext}, \eqref{prop:funext-characterizations:equiv} of \Cref{prop:funext-characterizations} relativized to \(\U\), this implies that \(\FE_{\U}\) holds and that \(\CEqToEq\) is an equivalence for all \(A,B \co \U\).
  From the latter, \(\CUA_{\U}\) follows by 2-out-of-3.
\end{proof}

The usual proof of \(\ITT + \UA_{\U} \vdash \FE_{\U}\) thus factors through a universe-independent proof of \(\ITT + \EI \vdash \FE\).
We now show that this decomposition of \(\UA_\U\) is proper: neither \(\CUA_\U\) nor \(\FE_\U\) alone implies \(\UA_\U\).
That \(\ITT + \FE_\U \nvdash \UA_\U\) follows from the standard model of type theory in \(\Set\), so it is \(\ITT + \CUA_\U \nvdash \UA_\U\) that we need to establish.
Equivalently, we want to show that \(\ITT + \CUA_\U \nvdash \FE_\U\).

\section{Von Glehn's polynomial model construction}\label{sec:von-glehn-polynomial}

We use Von Glehn's \emph{polynomial model} construction \cite{Glehn2015,MossGlehn2018} to separate \(\UA_{\U}\) from \(\CUA_{\U}\).
This model construction is related \cite[\S5.2]{Glehn2015} \cite[\S7.1]{Moss2018} to G\"{o}del's Dialectica interpretation \cite{Godel1958}, versions of which have been used to interpret various logics; see for example de Paiva's categorical models of linear logics \cite{DePaiva1991}.
For our purposes, the relevant feature of the polynomial models is that they always refute \(\FE\).

By a model of type theory, we mean a category with attributes \cite{Cartmell1978}, category with families \cite{Dybjer1996}, or natural model \cite{Awodey2016}: a presheaf of types \(\Ty[\cat{C}] \co \cat{C}^{\op} \to \Set\) and presheaf of terms \(\Tm[\cat{C}] \co (\elems[\cat{C}]\Ty[\cat{C}])^{\op} \to \Set\).
Here we depart from Von Glehn, who works with non-split models (categories with display maps), instead following Kov{\'a}cs' Agda formalization of a version of the construction \cite{Kovacs2020polynomial}.\footnote{
  Kov{\'a}cs constructs \(\Poly{\U}\) where \(\U\) is the model of type theory internal to Agda associated to some universe.
  Given a model \(\cat{A}\) of Agda's type theory, interpreting the formalization in \(\cat{A}\) yields the polynomial model whose base model is the interpretation of \(\U\) in \(\cat{A}\).
  Kov{\'a}cs postulates uniqueness of identity proofs and \(\FE\), but essentially only to deal with the mismatch between Agda's coproducts and the strict coproducts required for Von Glehn's construction.
  Unfortunately, uniqueness of identity proofs is inconsistent with \(\CUA_{\U}\), so we cannot build directly on this formalization.
}
We therefore describe the model in detail below, referencing Von Glehn's analogous construction for each component.
We leave it to the reader to translate some verifications to the split case.

We fix a base model \(\cat{C}\).
For \(\sigma \co \Delta \to \Gamma\), we write the action of \(\sigma\) on \(A \in \Ty[\cat{C}](\Gamma)\) and \(a \in \Tm[\cat{C}](\Gamma,A)\) as \(A\sigma\) and \(a\sigma\) respectively.
For \(a \in \Tm[\cat{C}](\Gamma,A)\), we write the induced substitution as \([a] \co \Gamma \to \Gamma.A\).
For \(\Gamma \in \cat{C}\), the category \(\Ty*[\cat{C}](\Gamma)\) has objects \(\Ty[\cat{C}](\Gamma)\) and morphisms \(A \to B\) given by \(\paren{\cat{C} / \Gamma}(\p_A, \p_B)\), i.e., morphisms \(\Gamma.A \to \Gamma.B\) over \(\Gamma\), i.e, elements of \(\Tm(\Gamma.A, B\p)\).
This extends to a functor \(\Ty*[\cat{C}] \co \cat{C}^{\op} \to \Cat\).
For ease of readability, we sometimes give variable names to context extensions, in which case we write \(\Gamma, a \co A\) instead of \(\Gamma.A\).

\begin{definition}
  The category \(\Poly{\cat{C}}\) is given by \(\elems[\Gamma \in \cat{C}] \Ty*[\cat{C}](\Gamma)^{\op}\), the Grothendieck construction of the fiberwise opposite of \(\Ty*[\cat{C}]\).
  For \(\Gamma \in \Poly{\cat{C}}\), we set \((\Shape\Gamma, \Pos\Gamma) \coloneqq \Gamma\) and refer to the two components as \emph{shapes} and \emph{positions} respectively.
  A morphism \(\sigma \co \Delta \to \Gamma\) is given by some \(\Shape{\sigma} \co \Shape{\Delta} \to \Shape{\Gamma}\) in \(\cat{C}\) and a morphism \(\Pos{\sigma} \co \Pos{\Gamma}\Shape{\sigma} \to \Pos{\Delta}\) in \(\Ty*[\cat{C}](\Shape{\Delta})\) \cite[(14)]{GambinoKock2013}.
  Composition of \(\sigma \co \Theta \to \Delta\) and \(\tau \co \Delta \to \Gamma\) is given by \(\Shape{(\sigma \cc \tau)} = \Shape{\sigma} \cc \Shape{\tau}\) and \(\Pos{(\sigma \cc\tau)} = \Pos{\tau} \cc \Pos{\sigma}\Shape{\tau}\), where \(\Pos{\sigma}\Shape{\tau} \co \Pos{\Theta}\Shape{\sigma}\Shape{\tau} \to \Pos{\Delta}\Shape{\tau}\).
\end{definition}

\begin{remark}
  Suppose \(\cat{C}\) is a democratic model of extensional type theory \cite[Definition 2.6]{ClairambaultDybjer2014}, so \(\Ty*[\cat{C}](\Gamma)\) is naturally equivalent to \(\cat{C}/\Gamma\).
  Then \(\Poly{\cat{C}}\) is the \emph{category of (single-variable) polynomials in \(\cat{C}\)}, denoted \(\mathrm{Poly}_{\cat{C}}(1,1)\) by Gambino and Kock \cite[\S2{\textperiodcentered}14]{GambinoKock2013}.
  \(\Poly{\cat{C}}\) is called the \emph{category of containers} in a line of work starting with Abbott, Altenkirch, and Ghani \cite{Abbott2003,AbbottAltenkirchGhani2003}.
  However, the \emph{container model} sketched by Altenkirch and Kaposi \cite{AltenkirchKaposi2021} is distinct from the polynomial model; it has the same category of contexts but more types.
\end{remark}

\begin{definition}
  The presheaf of types \(\Ty[\Poly{\cat{C}}]\) is given by restricting the presheaf \(\sum_{A \co \Ty} \Ty^{\Tm(A)}\) on \(\cat{C}\) along the projection \(\Poly{\cat{C}} \to \cat{C}\).
  This means an element of \(\Ty[\Poly{\cat{C}}](\Gamma)\) is given by some \(\Shape{A} \in \Ty[\cat{C}](\Shape\Gamma)\) and \(\Pos{A} \in \Ty[\cat{C}](\Shape\Gamma.\Shape{A})\), which we respectively refer to as the \emph{shapes} and \emph{positions} of \(A\).
\end{definition}

To make \(\Poly{\cat{C}}\) into a model, we need \(\cat{C}\) to have finite coproducts of types satisfying the strict \(\eta\) rule.
This is used in the substitution calculus (\cref{prop:polynomial-model:cwa}).
We can break these up into binary (\(A_0 + A_1\)) and nullary (0) coproducts.
We write the eliminator for binary coproducts as follows.
\begin{mathpar}
  \inferrule
  {\Gamma, v \co A_0 + A_1,\Delta(v) \vdash P(v)
    \and
    \text{for \(i \in \{0,1\}\): }\Gamma, a_i \co A_i,\Delta(\inc{i}(a_i)) \vdash u_i \co P(\inc{i}(a_i))}
  {\Gamma, v \co A_0 + A_1,\Delta(v) \vdash \elim_+^P(a_0.u_0, a_1.u_1, v) \co P(v)}
\end{mathpar}
The \(\beta\) rules state that \(\Gamma, a_i \co A_i, \Delta(\inc{i}(a_i)) \vdash \elim_+^P(a_0.u_0, a_1.u_1, \inc{i}(a_i)) \seq a_i\). 
The \(\eta\) rule states, as shown below, that we can test strict equality of terms depending on \(A_0 + A_1\) by checking equality on constructors.
\begin{mathpar}
  \inferrule
  {\Gamma, u \co A_0 + A_1, \Delta \vdash t(u), t'(u) \co P(u)
    \\
    \text{for \(i \in \{0,1\}\): }\Gamma, a \co A_i, \Delta(\inc{i}(a)) \vdash t(\inc{i}(a)) \seq t'(\inc{i}(a)) \co P(\inc{i}(a))}
  {\Gamma, u \co A_0 + A_1, \Delta \vdash t(u) \seq t'(u) \co P(u)}
\end{mathpar}
The elimination rule and \(\eta\) law for \(0\) are similar.
\begin{mathpar}
  \inferrule
  {\Gamma, v \co 0,\Delta(v) \vdash P(v)}
  {\Gamma, v \co 0,\Delta(v) \vdash \elim_0^P(v) \co P(v)}
  \and
  \inferrule
  {\Gamma, v \co 0,\Delta(v) \vdash t(v),t'(v) \co P(v)}
  {\Gamma, v \co 0,\Delta(v) \vdash t(v) \seq t'(v) \co P(v)}
\end{mathpar}

Semantically, the above rules mean that the split fibration \(\Ty*[\cat{C}]\) has split fibred coproducts~\cite[Definition~1.8.1]{Jacobs1999}: each \(\Ty*[\cat{C}](\Gamma)\) has chosen finite coproducts and the substitution functors \(\Ty*(\Gamma) \to \Ty*(\Delta)\) preserve them strictly.
For inference rules, see for example Von~Glehn~\cite[\S2.3.5]{Glehn2015} or Angiuli and Gratzer~\cite[\S2.5.1 and \S2.5.3]{AngiuliGratzer2025}.

\begin{remark}
  Strict \(\eta\) laws for coproducts are often omitted from the syntax of type theory due to issues with strict equality checking.
  In the simply-typed \(\lambda\)-calculus, strict equality checking for coproducts with the \(\eta\) law is decidable but difficult.
  Ghani \cite{Ghani1995} addresses the case of binary coproducts; Scherer \cite{Scherer2017} handles the empty type.
  In \(\ITT\), the \(\eta\) law for the empty type makes strict equality undecidable: with this law, deciding whether \(a \co A \vdash \inc{0}(\star) \seq \inc{1}(\star) \co 1 + 1\) requires deciding whether \(A\) implies \(0\).
  To our knowledge, it is an open problem whether strict equality is decidable for \(\ITT\) with binary coproducts and their \(\eta\) law; see, for example, discussion between Shulman, Kov{\'a}cs, and others on Proof Assistants StackExchange \cite{ShulmanKovacs2022}.
\end{remark}

We further require that our coproducts are \emph{extensive}.
The syntactic counterpart is often called \emph{large elimination}.
This means that given types \(\Gamma, a_i \co A_i, \Delta(\inc{i}(a_i)) \vdash P_i(a_i)\) for \(i \in \{0,1\}\), there is a type \(\Gamma, u \co A_0+A_1,\Delta(u) \vdash [P_0,P_1](u)\) satisfying \(\Gamma, a_i \co A_i,\Delta(\inc{i}(a_i)) \vdash [P_0,P_1](\inc{i}(a_i)) \seq P_i(a_i)\).
Using the strict \(\eta\) rule, for every family \(\Gamma. A_0 + A_1 \vdash P\) there is a canonical strict isomorphism \(P \sequiv [P\inc{0}, P\inc{1}]\).
Furthermore, for all such \(P\) there is a strict isomorphism \(\sum_{u \co A + B} P(u) \sequiv \sum_{a \co A} P(\inc{0}(a)) + \sum_{b \co B} P(\inc{1}(b))\).
This is used in particular in the construction of dependent product types (\cref{prop:polynomial-model:pi-types}).

For the rest of this section, we assume \(\cat{C}\) has extensive finite coproducts of types with the strict \(\eta\) rule.

\begin{definition}\label[definition]{def:polynomial-model:sections}
  The presheaf of terms \(\Tm[\Poly{\cat{C}}] \co \paren{\smallint\Ty[\Poly{\cat{C}}]}^{\op} \to \Set\) is given by \(\Tm[\Poly{\cat{C}}](\Gamma, A) \coloneqq \sum_{\Shape{a} \co \Tm(\Shape{\Gamma}, \Shape{A})} \Ty*(\Shape{\Gamma})(\Pos{A}[\Shape{a}], \Pos{\Gamma})\).
  We refer to the components again as \emph{shapes} and \emph{positions}.
\end{definition}

\begin{proposition}\label[proposition]{prop:polynomial-model:cwa}
  \(\paren{\Poly{\cat{C}}, \Ty[\Poly{\cat{C}}], \Tm[\Poly{\cat{C}}]}\) extends to a category with families by setting
  \begin{mathpar}
    \Shape{(\Gamma.A)} \coloneqq \Shape\Gamma.\Shape{A},
    \and
    \Pos{(\Gamma.A)} \coloneqq \Pos{\Gamma}\p_{\Shape{A}} + \Pos{A},
    \and
    \Shape{(\p_A)} \coloneqq \p_{\Shape{A}},
    \and
    \Pos{(\p_A)} \coloneqq {\inc0} \co \Pos\Gamma \to \Pos\Gamma\p + \Pos A, \\
    \Shape{(\q_A)} \coloneqq \q_{\Shape{A}},
    \and
    \Pos{(\q_A)} \coloneqq \inc{1} \co \Pos{A} \to \Pos\Gamma\p + \Pos{A},
    \and
    \Shape{\pair{\sigma,a}} \coloneqq \pair{\Shape\sigma, \Shape{a}},
    \and 
    \Pos{\pair{\sigma, a}} \coloneqq [\Pos\sigma, \Pos{a}].
  \end{mathpar}
\end{proposition}
\begin{proof}
  Given \(\sigma \co \Delta \to \Gamma\) and \(a \in \Tm(\Delta, A\sigma)\) we have \(\Shape{\pair{\sigma,a}} \co \Shape\Delta \to \Shape\Gamma.\Shape A\) and \(\Pos{\pair{\sigma, a}} \co \Pos{\Gamma}\Shape\sigma + \Pos{A}\pair{\Shape\sigma, \Shape a} \to \Pos{\Delta}\) in \(\Ty*(\Delta)\).
  Clearly, all desired equations hold on shapes since they hold in \(\cat{C}\).
  We have \(\Pos{\paren{ \p\pair{\sigma, a} }} = [\Pos\sigma, \Pos{a}] \cc \inc{0}\p = \Pos{\sigma}\) and \(\Pos{\paren{ \q\pair{\sigma, a} }} = [\Pos\sigma, \Pos{a}] \cc \inc{1}\p = \Pos{a}\).
  This shows \(\p\pair{\sigma,a} = \sigma\) and \(\q\pair{\sigma,a} = a\).
  Furthermore, \(\Pos{\pair{\p, \q}} = [\Pos{\p}, \Pos{\q}] = [\inc{0}, \inc{1}] = \id\) and \(\Pos{\paren{\pair{\sigma, a}\tau}} = \Pos{\tau} \cc [\Pos{\sigma}, \Pos{a}]\Shape{\tau} = [\Pos\tau \cc \Pos{\sigma}\Shape{\tau}, \Pos\tau \cc \Pos{a}\Shape{\tau}] = \Pos{\pair{\sigma\tau, a\tau}}\).
  This shows \(\pair{\p,\q} = \id\) and \(\pair{\sigma, a}\tau = \pair{\sigma\tau, a\tau}\).
\end{proof}

\subsection{Type formers}

We give the interpretations in \(\Poly{\cat{C}}\) of \(\Sigma\), identity, \(\Pi\), binary and nullary coproduct, and universe types.

\begin{proposition}[cf.~{\cite[\S4.2]{Glehn2015}}]\label[proposition]{prop:polynomial-model:sigma-types}
  Given \(A \in \Ty[\Poly{\cat{C}}](\Gamma)\), \(B \in \Ty[\Poly{\cat{C}}](\Gamma.A)\), we have a dependent sum \(\Sigma_{A} B \in \Ty(\Gamma)\) given by
  \[
    \Shape{\Gamma} \vdash \Shape{\paren{\Sigma_AB}} \coloneqq \sum_{a \co \Shape{A}} \Shape{B}(a),
    \qquad
    \Shape{\Gamma}, \pair{\Shape{a},\Shape{b}} \co \Shape{\paren{\Sigma_AB}} \vdash \Pos{\paren{\Sigma_AB}} \coloneqq \Pos{A}(\Shape{a}) + \Pos{B}(\Shape{a},\Shape{b}).
  \]
\end{proposition}

\begin{proposition}[{cf.~\cite[\S4.4]{Glehn2015}}]\label[proposition]{prop:polynomial-model:identity-types}
  Given \(A \in \Ty[\Poly{\cat{C}}](\Gamma)\) and \(u, v \in \Tm(\Gamma, A)\), we have an identity type \(u =_A v \in \Ty(\Gamma)\) given by
  \[
    \Shape\Gamma \vdash \Shape{\paren{u =_A v}} \coloneqq \paren{\Shape{u} =_{\Shape{A}} \Shape{v}},
    \qquad
    \Shape\Gamma,p \co \Shape{\paren{u =_A v}} \vdash \Pos{\paren{u =_A v}} \coloneqq 0,
  \]
  with, for \(u \in \Tm(\Gamma, A)\), the reflexive path \(\refl[u] \in \Tm(\Gamma, u =_A u)\) given by
  \[
    \Shape\Gamma \vdash \Shape{(\refl[u])} \coloneqq \refl[\Shape{u}] \co \paren{\Shape{u} =_{\Shape{A}} \Shape{u}},
    \qquad
    \Shape\Gamma.0 \vdash \Pos{(\refl[u])} \coloneqq \elim_0 \co \Pos\Gamma.
  \]
\end{proposition}
\begin{proof}
  For every \(u \in \Tm(\Gamma, A)\), \(B \in \Ty(\Gamma.A. u = \q)\), \(w \in \Tm(\Gamma, B[u, \refl[u]])\), \(v \in \Tm(\Gamma, A)\), and \(p \in \Tm(\Gamma, u =_A v)\), the eliminator \(\elim_=^{B,u}(w,v,p) \in \Tm(\Gamma, B[v, p])\) is given by
  \begin{alignat*}{3}
    &\Shape\Gamma &&\vdash \Shape{\elim_=^{B, u}(w, v, p)} \coloneqq \elim_=^{\Shape{B}, \Shape{u}}(\Shape{w}, \Shape{v}, \Shape{p}) \co \Shape{B}[\Shape{v}, \Shape{p}],  \\
    &\Shape\Gamma, b \co \Pos{B}[\Shape{v}, \Shape{p}, \Shape{\elim_=^{B, u}(w, v, p)}]  &&\vdash \Pos{\elim_=^{B, u}(w, v, p)} \coloneqq\Pos{w}(p_*b) \co \Pos{\Gamma},
  \end{alignat*}
  where \(\Shape{\Gamma}.\Pos{B}[\Shape{u}, \refl[\Shape{u}], \Shape{w}] \vdash \Pos{w} \co \Pos{\Gamma}\) and \(p_* \co \Pos{B}[\Shape{v}, \Shape{p}, \Shape{\elim_=^{B, u}(w, v, p)}] \to \Pos{B}[\Shape{u}, \refl[\Shape{u}], \Shape{w}]\) is defined by path induction on \(\Shape{p}\).
  The \(\beta\) rule for \(\elim_=\) follows from the same \(\beta\) rule in the base model.
\end{proof}

\begin{remark}
  Von Glehn takes \(\Pos{\paren{u =_A v}}\) to be the constant family \(\Pos{A}(\Shape{u}) + \Pos{A}(\Shape{v})\).
  We follow Kov{\'a}cs \cite{Kovacs2020polynomial} by instead taking the constant family \(0\), which simplifies our arguments.
  Since both definitions satisfy the rules of the identity type, they are equivalent---though not categorically equivalent.
  The equivalence suffices to imply that our main result \Cref{prop:poly-inherits-famcua} transfers to Von Glehn's identity types.
  This kind of flexibility is common to type formers without strict \(\eta\) laws in the polynomial model.
\end{remark}
We treat the construction of \(\Pi\) types in more detail.
Here, we depend crucially on extensivity of coproducts.
First, we define families over \(A_0 + A_1\) that are inhabited over exactly one of the inclusions.

\begin{definition}
  For types \(A_0, A_1\), define the families \(A_0 + A_1 \vdash \isInc{0} \coloneqq [1,0]\) and \(A_0 + A_1 \vdash \isInc{1} \coloneqq [0,1]\).
\end{definition}

\begin{lemma}\label[lemma]{prop:is-inc-is-strict-prop}
  For types \(A_0, A_1\), the families \(A_0 + A_1 \vdash  \isInc{0}, \isInc{1}\) are strict propositions.
\end{lemma}
\begin{proof}
  More generally, if \(A_i \vdash P_i\) are strict propositions for \(i \in \{0,1\}\) then so is \([P_0,P_1]\): we have  \(u \co A_0 + A_1, v_0, v_1 \co [P_0, P_1] \vdash v_0 \seq v_1\) directly from the strict \(\eta\) rule  since \([P_0, P_1](\inc{i}(u_i)) \seq P_i(u_i)\).
\end{proof}

\begin{proposition}[cf.~{\cite[\S4.3]{Glehn2015}}]\label[proposition]{prop:polynomial-model:pi-types}
  Given \(A \in \Ty[\Poly{\cat{C}}](\Gamma)\), \(B \in \Ty[\Poly{\cat{C}}](\Gamma.A)\), we have a dependent product \(\Pi_{A} B \in \Ty(\Gamma)\) given by
  \begin{alignat*}{3}
    &\Shape{\Gamma} &&\vdash \Shape{\paren{\Pi_AB}} &&\coloneqq \sum_{\Shape{f} \co \prod_{a \co \Shape{A}} \! \Shape{B}(a)} \prod_{\;a \co \Shape{A}} \Pos{B}(a, \Shape{f}(a)) \to 1 + \Pos{A}(a), \\
    &\Shape{\Gamma}, \pair{\ShapeShape f,\ShapePos {f}} \co \Shape{\paren{\Pi_AB}} &&\vdash \Pos{\paren{\Pi_AB}} &&\coloneqq \sum_{\substack{ a \co \Shape{A} \\ b \co \Pos{B}(a, \ShapeShape{f}(a)) }} {\isInc{0}}(\ShapePos{f}(a, b)).
  \end{alignat*}
  For \(f \in \Tm(\Gamma, \Pi_A B)\) we write \(\ShapeShape{f}\) and \(\ShapePos{f}\) for the first and second component of \(\Shape{f}\) respectively.
\end{proposition}
\begin{proof}
  It suffices to define a natural isomorphism \(\lambda \co \Tm(\Gamma.A, B) \cong \Tm(\Gamma, \Pi_A B) \oc \app\).
  By \Cref{def:polynomial-model:sections}, an element of \(\Tm(\Gamma.A, B)\) corresponds to a pair
  \begin{mathpar}
    \Shape\Gamma \vdash \Shape{b} \co \prod_{a \co \Shape{A}} \Shape{B}(a),
    \qquad
    \Shape\Gamma \vdash \Pos{b} \co \prod_{a \co \Shape{A}} \Pos{B}(a, \Shape{b}(a)) \longrightarrow \Pos{\Gamma} + \Pos{A}(a).
  \end{mathpar}
  By \Cref{def:polynomial-model:sections} and the curry-uncurry isomorphism, an element of \(\Tm(\Gamma, \Pi_A B)\) corresponds to a triple
  \begin{mathpar}
    \Shape\Gamma \vdash \ShapeShape{{{f}}} \co \prod_{a \co \Shape{A}} \Shape{B}(a),
    \and
    \Shape\Gamma \vdash \ShapePos{{{f}}} \co \prod_{a \co \Shape{A}} \Pos{B}\paren[\big]{a, \ShapeShape{{{f}}}(a)} \longrightarrow 1 + \Pos{A}(a),
    \\
    \Shape\Gamma \vdash \Pos{f} \co \prod_{a \co \Shape{A}} \paren[\bigg]{ \sum_{b \co \Pos{B}\paren[]{a, \ShapeShape{{{f}}}(a)} } {\isInc{0}}(\ShapePos{{{f}}}(a, b)) } \longrightarrow \Pos\Gamma.
  \end{mathpar}
  Now, note that we have for all types \(X, Y, Z\) that
  \[
    \sum_{f \co X \to 1 + Y} \prod_{x \co X} Z^{\isInc{0}(fx)}
    \sequiv \prod_{x \co X} \sum_{u \co 1 + Y} Z^{\isInc{0}(u)} 
    \sequiv \prod_{x \co X} \paren[\bigg]{ \sum_{\star\co 1} Z^{\isInc{0}(\inc{0}(\star))} + \sum_{y \co Y} Z^{\isInc{0}(\inc{1}(y))} }
    \sequiv \paren{Z + Y}^{X}.
  \]
  Applying this strict isomorphism to the above yields the desired bijection.
\end{proof}

\begin{proposition}[cf.~{\cite[\S4.2]{Glehn2015}}]\label[proposition]{prop:polynomial-model:finite-coproducts}
  We have a nullary coproduct \(0 \in \Ty[\Poly{\cat{C}}]\) and binary coproduct \(A + B \in \Ty(\Gamma)\) for \(A,B \in \Ty[\Poly{\cat{C}}](\Gamma)\) given by
  \begin{alignat*}{6}
    &\Shape{\Gamma} &&\vdash \Shape{\paren{A + B}} &&\coloneqq \Shape{A} + \Shape{B},
    &\qquad\qquad\qquad&\Shape{\Gamma} &&\vdash \Shape{0} &&\coloneqq 0, \\
    &\Shape{\Gamma}.(\Shape{A} + \Shape{B}) &&\vdash \Pos{\paren{A + B}} &&\coloneqq [\Pos{A},\Pos{B}],
    &\qquad\qquad\qquad&\Shape{\Gamma}.0 &&\vdash \Pos{0} &&\coloneqq \elim_0 .
  \end{alignat*}
  These satisfy the strict \(\eta\) rule and are extensive.
\end{proposition}

\begin{proposition}[cf.~{\cite[Proposition 7.1.6]{Moss2018}}]\label[proposition]{prop:polynomial-model:universe}
  We have a universe \(\U \in \Ty[\Poly{\cat{C}}](1)\) with decoding function \(\El \in \Ty[\Poly{\cat{C}}](1.\U)\) (which we usually leave implicit) given by
  \begin{alignat*}{6}
    &1 &&\vdash \Shape{\U} &&\coloneqq \sum_{\Shape{A} \co \U} \U^{\Shape{A}},
    &\qquad\qquad\qquad &1, \pair{\Shape{A},\Pos{A}} \co \Shape{\U} &&\vdash \Shape{\El} &&\coloneqq \Shape{A}, \\
    & 1, \pair{\Shape{A},\Pos{A}} \co \Shape{\U} &&\vdash \Pos{\U} &&\coloneqq 0,
    &\qquad\qquad\qquad &1, \pair{\Shape{A},\Pos{A}} \co \Shape{\U}, \Shape{a} \co \Shape{\El}\pair{\Shape{A},\Pos{A}} &&\vdash \Pos{\El} &&\coloneqq \Pos{A}(\Shape{a}).
  \end{alignat*}
\end{proposition}

\subsection{A dependent right adjoint}

The category \(\cat{C}\) is a reflective subcategory of \(\Poly{\cat{C}}\): the functor \(\Shape{-} \co \Poly{\cat{C}} \to \cat{C}\) has a fully faithful right adjoint \(\Mod \co \cat{C} \to \Poly{\cat{C}}\) given on objects by \(\Mod\Gamma \coloneqq (\Gamma, 0)\). 
Von Glehn~\cite{Glehn2015} uses this adjunction to define the model that we defined manually above, transferring it from \(\cat{C}\).
Both functors extend to pseudomorphisms of models in the sense of Kaposi, Huber, and Sattler \cite[\S4]{KaposiHuberSattler2019}.

\begin{lemma}
  The adjunction \(\Mod \co \cat{C} \leftrightarrows \Poly{\cat{C}} \, \colon\! \Shape{-}\) lifts to an adjunction of pseudomorphisms of models, with left adjoint projecting to shapes of types and terms, and the right adjoint given on types and terms by \(\Mod A \coloneqq (A, 0)\) and \(\Mod a \coloneqq (a, \elim_0)\).
\end{lemma}
\begin{proof}
  Immediate from the definition of context extension and that \(0 + 0 \sequiv 0\).
\end{proof}

The right adjoint morphism induces a dependent right adjoint (cf.~\cite[\S7]{GratzerKavvosNuytsBirkedal2021}).

\begin{corollary}
  The operation \(\Ty(\Shape\Gamma) \to \Ty(\Gamma), A \mapsto \paren{\Mod A}\eta_\Gamma\) defines a dependent right adjoint.
\end{corollary}
\begin{proof}
  \(\Tm(\Gamma, \paren{\Mod A}\eta_\Gamma) \cong (\Poly{\cat{C}}/{\Mod \Gamma})(\Gamma, \Mod\Shape\Gamma.{\Mod A}) \cong (\cat{C}/\Shape\Gamma)(\Shape\Gamma, \Shape\Gamma.A) \cong \Tm(\Shape{\Gamma}, A)\).
\end{proof}

Henceforth, we denote by \(\Mod\) the dependent right adjoint, not the morphism.
The composite mapping \(A \mapsto \Mod(\Shape A)\) defines a pointed endofunctor on \(\Ty[\Poly{\cat{C}}](\Gamma)\) (and a lex operation in the sense of~\cite[Remark 5]{CoquandRuchSattler2021}).
If clear from context, we write just \(\Mod A\) for this composite.

\section{Familial categorical univalence in the polynomial model}\label{sec:categorical-univalence}

Fix an input model \(\cat{C}\) as in \cref{sec:von-glehn-polynomial}.
To study \(\FamCUA_{\U}\) in \(\Poly{\cat{C}}\), we analyze the wild category \(\U^I\) and its isomorphisms.
To simplify calculations, we redefine here \(\U^I(A, B) \coloneqq \prod_{u \co \sum_I A} B(\pi_0 u)\).
This is strictly isomorphic to the type \(\prod_{i \co I} A(i) \to B(i)\) in \Cref{defn:familial-categorical-univalence}, so the two versions of \(A \catIso[\U^I] B\) are related by an equivalence preserving the identity isomorphism up to path.
Hence, \(\FamCUA_{\U}\) is invariant under this change.

We start by unfolding the type \(\U^I(A, B)\) and composition in \(\U^I\).
A key observation is that the shape part of \(f \in \Tm(\Gamma, A \to B)\) consists of a function between shapes and a \emph{partial} function between positions.

\begin{lemma}\label[lemma]{prop:characterization-of-families-of-functions}
  For \(\Gamma \in \Poly{\cat{C}}\), \(I \in \Ty(\Gamma)\), and \(A, B \in \Ty(\Gamma.I)\), the type \(\U^I(A, B) \in \Ty(\Gamma)\) is given by
  \begin{alignat*}{3}
    &\Shape{\Gamma} &&\vdash \Shape{\U^I(A, B)} \seq \sum_{ \Shape{f} \co \U^{\Shape{I}}(\Shape A,\Shape B) } \prod_{\substack{ i \co \Shape I \\ a \co \Shape A(i) }} \Pos{B}(i, \Shape{f}(i,a)) \to 1 + \paren[\big]{ \Pos{I}(i) + \Pos{A}(i, a)  }, \\
    &\Shape{\Gamma}, \pair{\Shape{f}, \Pos{f}} \co \Shape{\U^I(A, B)} &&\vdash \Pos{{\U^I(A, B)}} \seq \sum_{ \substack{ i \co \Shape{I}, \, a \co \Shape{A}(i) \\ b \co \Pos{B}(i, a, \Shape{f}(a)) } } \isInc{0}(\Pos{f}(i, a, b))
  \end{alignat*}
\end{lemma}
\begin{proof}
  Direct unfolding using \Cref{prop:polynomial-model:sigma-types,prop:polynomial-model:pi-types}.
\end{proof}

\begin{lemma}\label[lemma]{prop:characterization-of-shapes-of-composite}
  If \(f \in \Tm(\Gamma, \U^I(B, C))\), \(g \in \Tm(\Gamma, \U^I(A, B))\), then the composite \(fg \in \Tm(\Gamma, \U^I(A, C))\) is given by
  \begin{alignat*}{3}
    \Shape\Gamma &\vdash \ShapeShape{\paren{fg}} &&\seq \ShapeShape{f} \cc \ShapeShape{g} \co \U^{\Shape{I}}(\Shape{A}, \Shape{C}),
    \\
    \Shape\Gamma &\vdash \ShapePos{\paren{fg}} &&\seq \lambda i.\lambda a.[\inc{0}, \inc{0},\ShapePos{g}(i,a)] \cc \ShapePos{f}(\ShapeShape{g}(i,a)) \co \smashoperator{ \prod_{\substack{ i \co \Shape I \\ a \co \Shape{A}(i) }} } \Pos{C}(i, \ShapeShape{\paren{fg}}(a)) \to 1 + \paren[\big]{ \Pos{I}(i) + \Pos{A}(i, a) }.
  \end{alignat*}
\end{lemma}
\begin{proof}
  We have that \(\Tm(\Gamma, \U^{I}(A, B)) \cong \Tm(\Gamma.I.A, B\p)\).
  Let \(u \in \Tm(\Gamma.I.B, C\p)\) and \(v \in \Tm(\Gamma.I.A, B\p)\) given by \(\Shape{u} \in \Tm(\Shape\Gamma.\Shape{I}.\Shape{B}, \Shape C\p)\), \(\Shape v \in \Tm(\Shape \Gamma.\Shape I.\Shape A, \Shape B\p)\), \(\Pos{u} \co \Pos{C}\pair{\p, \Shape{u}} \to \paren{ \Pos{\Gamma} + \Pos{I} }\p + \Pos{B}\) in \(\Ty*(\Shape\Gamma.\Shape I.\Shape{B})\), and \(\Pos{v} \co \Pos{B}\pair{\p, \Shape{v}} \to \paren{ \Pos{\Gamma} + \Pos{I} }\p + \Pos{A}\) in \(\Ty*(\Shape\Gamma.\Shape{I}.\Shape{A})\).
  The composite of \(u\) and \(v\) is by definition \(u\pair{\p, v} \in \Tm(\Gamma.I.A, C\p)\).
  Direct calculation using \Cref{def:polynomial-model:sections,prop:polynomial-model:cwa} shows that \(\Shape{\paren{u\pair{\p, v}}} = \Shape u\pair{\p, \Shape v}\) and \(\Pos{\paren{u\pair{\p, v}}} = [\inc{0}, \Pos{v}] \cc \Pos{u}\pair{\p, \Shape{v}} \co \Pos{C}\Shape u\pair{\p, \Shape v} \to \paren{ \Pos{\Gamma} + \Pos{I} }\p + \Pos{A}\).
  Composing with the \(\lambda\)-\(\app\) bijection from \Cref{prop:polynomial-model:pi-types} yields the desired description.
\end{proof}

\subsection{Categories of partial functions}

We now introduce an auxiliary wild category in \(\cat{C}\).
It can be viewed as the Kleisli category of the monad on \(\U^I\) given by coproduct with a fixed family \(J \co I \to \U\), though we will not explicitly develop this viewpoint.
To see that this even is a wild category in our setting, we rely on the strict properties of coproducts.

\begin{proposition}[In \(\cat{C}\)]
  For every family \(J \co I \to \U\), the following defines a wild category \(\U_J^I\):
  \begin{gather*}
    \paren{\U_J^I}_0 \coloneqq \U^I,
    \quad
    \paren{\U_J^I}_1(A, B) \coloneqq \prod_{i \co I} A(i) \to J(i) + B(i),
    \quad
    \paren{\id_A}_i \coloneqq \inc{1},
    \quad
    \paren{f \cc g}_i \coloneqq [\inc{0}, f_i] \cc g_i,
  \end{gather*}
  with unitors and associators given by reflexivity.
\end{proposition}
\begin{proof}
  Direct calculation using the \(\eta\) rules for \(\Pi\) and \(+\).
\end{proof}

Morphisms in \(\U_J^I\) can be thought of as families of partial functions, with \(J\) as a type of ``errors''.
We introduce a notion of \emph{total} morphism in \(\U_J^I\).
By \(\eta\) for coproducts, total morphisms coincide with morphisms in \(\U^I\) up to equivalence.
Crucially, all isomorphisms in \(\U^I_J\) will be total.

\begin{definition}\textbf{(In \({\cat{C}}\))}
  A morphism \(f \co \U^I_J(A, B)\) is \emph{total} if \(\isTot(f) \coloneqq \prod_{i\co I, a \co A(i)} \isInc{1}(f_ia)\) is inhabited.
  We define \(\U^I_{J, \tot}(A, B) \coloneqq \sum_{f \co \U^I_J(A, B)} \isTot(f)\).
\end{definition}

\begin{lemma}[In \({\cat{C}}\)]\label[lemma]{prop:being-total-is-proposition}
  For all \(f \co \U^I_J(A, B)\), the type \(\isTot(f)\) is a homotopy proposition.
\end{lemma}
\begin{proof}
 \(\Pi_CP\) is a strict proposition if \(P\) is: for \(p, q \co \Pi_CP\) we have \(p \seq \lambda x.p(x) \seq \lambda y.q(y) \seq q\).
 Hence, it follows from \Cref{prop:is-inc-is-strict-prop} that \(\isTot(f)\) is even a strict proposition.
\end{proof}

\begin{lemma}[In \({\cat{C}}\)]\label[lemma]{prop:families-and-total-families-lemma}
  For all types \(I\) and \(I \vdash J, B\), the map \( \paren{\prod_{i \co I} B(i)} \to \prod_{i \co I} \sum_{ u \co J(i) + B(i) } \isInc{1}(u)\) given by \(f \mapsto \lambda i.(\inc{1}(f_iu), \star) \) is an equivalence.
\end{lemma}
\begin{proof}
  \(\Pi\) type formation sends families of strict isomorphisms to strict isomorphisms.
  For \(i \co I\) we have
  \[
    B(i)
    \sequiv \paren[\bigg]{ \sum_{j \co J(i)} 0 } + \paren[\bigg]{ \sum_{b \co B(i)} 1 }
    \sequiv \paren[\bigg]{ \sum_{j \co J(i)} \isInc{1}(\inc{0}(j)) } + \paren[\bigg]{ \sum_{b \co B(i)} \isInc{1}(\inc{1}(b)) }
    \sequiv \sum_{u \co J(i) + B(i)} \isInc{1}(u).
  \]
  In each step we use that to check the commutation out of a coproduct, it suffices to check after precomposing with both inclusions.
\end{proof}
\begin{corollary}[In \({\cat{C}}\)]\label[corollary]{prop:families-of-functions-and-total-families-of-functions-coincide}
  For all types \(I\) and \(J, A, B \co I \to \U\), the map \(\U^I(A, B) \to \U^I_{J,\tot}(A, B)\) given by \(f \mapsto (\inc{1} \cc f, \lambda i. \lambda a. \star)\) is an equivalence.
\end{corollary}
\begin{proof}
  Instantiate \Cref{prop:families-and-total-families-lemma} with index type \(\sum_{i \co I} A\) and the families \((i, a) \co \sum_{I} A \vdash J(i), B(i)\).
  The result follows by composing with the strict curry-uncurry isomorphism.
\end{proof}

\begin{lemma}[In \(\cat{C}\)]\label[lemma]{prop:total-function-left-cancel}
  Given a pair of morphisms \(f \co \U^I_J(B, C)\), \(g \co \U^I_J(A, B)\), if \(f \cc g\) is total then so is \(g\).
\end{lemma}
\begin{proof}
  A morphism \(h \co \U^I_J(A, B)\) is total if and only if \(\prod_{i\co I, a \co A(i)} \isInc{0}(h_ia) \to 0\).
  For \(i \co I\), \(a \co A(i)\) we have \(\isInc{0}(g_ia) \to \isInc{0}((f \cc g)_i (a))\) by the definition of \(\cc\), and so if \(f \cc g\) is total we get \(\isInc{0}(g_ia) \to 0\).
\end{proof}

\begin{corollary}[In \(\cat{C}\)]\label[corollary]{prop:partial-equivalences-are-total}
  All isomorphism in \(\U^I_J\) are total.
\end{corollary}
\begin{proof}
  By induction, totality transfers along paths.
  Hence, the claim follows since \(\id\) is total.
\end{proof}

\begin{proposition}[In \(\cat{C}\)]\label[proposition]{prop:type-of-partial-equivalences-is-equivalent-to-equivalences}
    For all types \(I \co \U\) and families \(J, A, B \co \U^{I}\) we have \(\paren{A \catIso[\U^{I}] B} \simeq \paren{A \catIso[\U^{I}_J] B}\).
\end{proposition}
\begin{proof}
  We have a chain of maps \(u \co \U^I(A, B) \to \U^I_{J,\tot}(A, B) \to \U_J^I(A, B)\).
  The map \(u\) strictly preserves identities and composition and therefore lifts to subtypes of isomorphisms (recall that being an isomorphism is a proposition by \Cref{prop:being-isomorphism-is-proposition}) via \(v \co (A \catIso[\U^I] B) \to (A \catIso[\U_J^I] B)\), \(\pair{f, s, S, r, R} \mapsto \pair{uf, us, \ap{u}S, ur, \ap{u}R}\).
  Our goal is to show that this restriction is an equivalence.
  The fibers of \(u\) (and thus also \(v\)) are propositions, since the first component of \(u\) is an equivalence by \Cref{prop:families-of-functions-and-total-families-of-functions-coincide} and the second component is an embedding by \Cref{prop:being-total-is-proposition}.
  Hence, \(\ap{u} \co \paren{ f =_{\U^I(A, B)} g } \to \paren{ uf =_{\U^I_{J}(A, B)} ug }\) is an equivalence for all \(f, g \co \U^I(A, B)\).
  The fibers of \(u\) are inhabited over isomorphisms and their sections and retractions by \Cref{prop:partial-equivalences-are-total} and the fact that sections and retractions of isomorphisms are isomorphisms.  
  Thus, the fibers of \(v\) are inhabited.
\end{proof}

\subsection{Familial categorical univalence}
\label{sec:familial-categorical-univalence}

To verify that \(\Poly{\cat{C}}\) inherits \(\FamCUA_{\U}\), we analyze the wild category \(\U^I\) in this model.
In \(\Poly{\cat{C}}\), we have for \(I \co \U\) the type \(\U^I\) of \(I\)-indexed families.
Over \(A, B \co \U^I\), we have the type \(A \CEq_{\U^I} B\) of isomorphisms between them.
We analyze the shapes of these types (i.e.\ the image under \(\Shape{-}\)) in the base model \(\cat{C}\).
For clarity, we use different notation: define \(I \co \U \vdash \Fam(I) \coloneqq \U^I\) and \(I \co \U, A, B \co \Fam(I) \vdash \CEq*(I, A, B) \coloneqq (A \CEq_{\U^I} B)\).
Now \(\Shape{\U}\) is a closed type of \(\cat{C}\), \(\Shape{\Fam}\) is a family of types over it, \(\Shape{\CEq*}\) is a family over \(\Shape{\U}\) and two copies of \(\Shape\Fam\), and \(\Pos{\CEq*}\) is a family over \(\Shape{\CEq*}\).

\begin{remark}
  Note that the following data is more ordered than it might seem at first.
  \eqref{prop:characterization-of-Familial-categorical-equivalence-in-polynomial-model:shape-shape} is exactly the data of an isomorphism in the wild category \(\U^{\Shape{I}}\).
  \eqref{prop:characterization-of-Familial-categorical-equivalence-in-polynomial-model:shape-pos} is the data of an isomorphism in the wild category \(\U_{K}^{J}\) for some \(J, K\), modulo the first equivalence.
  Viewing the morphisms in this wild category again as partial functions, the data given by~\eqref{prop:characterization-of-Familial-categorical-equivalence-in-polynomial-model:pos} are exactly the inputs on which the functions are not defined.
\end{remark}

\begin{lemma}\label[lemma]{prop:characterization-of-Familial-categorical-equivalence-in-polynomial-model}
  Let \(I \in \Tm[\Poly{\cat C}](\Gamma,\U)\) and \(A, B \in \Tm[\Poly{\cat{C}}](\Gamma, \Fam(I))\).
  The shapes of the type \(\CEq*(I, A, B) \in \Ty[\Poly{\cat{C}}](\Gamma)\) are equivalent to the iterated \(\Sigma\) type in context \(\Shape\Gamma\) given by the following:
  \begin{equation}\begin{gathered}\label{prop:characterization-of-Familial-categorical-equivalence-in-polynomial-model:shape-shape}
    s \co \prod_{i \co \Shape{I}} \Shape{B}(i) \longrightarrow \Shape{A}(i),
    \qquad
    f \co \prod_{i \co \Shape{I}} \Shape{A}(i) \longrightarrow \Shape{B}(i),
    \qquad
    r \co \prod_{i \co \Shape{I}} \Shape{B}(i) \longrightarrow \Shape{A}(i),
    \\
    S \co {f} \cc {s} = \id \quad\text{in}\quad \U^{\Shape{I}}(\Shape{B},\Shape{B}),
    \qquad
    R \co {r} \cc {f} = \id \quad\text{in}\quad \U^{\Shape{I}}(\Shape{A},\Shape{A}),
  \end{gathered}\end{equation}
  as well as the following functions and paths over this equivalence
  \begin{equation}\begin{gathered}\label{prop:characterization-of-Familial-categorical-equivalence-in-polynomial-model:shape-pos}
    \widetilde{f} \co {\prod_{\substack{ i \co \Shape{I} \\ a \co \Shape{A}(i) }}} \Pos{B}({f}(i,b)) \longrightarrow \paren[\big]{1 + \Pos{I}(i)} + \Pos{A}(i,a),
    \\
    \widetilde{s} \co \smashoperator{\prod_{\substack{ i \co \Shape{I} \\ b \co \Shape{B}(i) }}} \Pos{A}({s}(i,b)) \longrightarrow \paren[\big]{1 + \Pos{I}(i)} + \Pos{B}(i,b),
    \quad
    \widetilde{r} \co \smashoperator{\prod_{\substack{ i \co \Shape{I} \\ a \co \Shape{B}(i) }}} \Pos{A}({r}(i,b)) \longrightarrow \paren[\big]{1 + \Pos{I}(i)} + \Pos{B}(i,b),
    \\
    S_*\paren[\big]{\widetilde{s} \cc \paren{\widetilde{f}{s}} } = \id \quad\text{in}\quad \U^{\sum_{\Shape{I}}\!\Shape{B}}_{1 + \Pos{I}}(\Pos{B}, \Pos{B}),
    \qquad
    R_*\paren[\big]{\widetilde{f} \cc \paren{\widetilde{r}{f}} } = \id \quad\text{in}\quad \U^{\sum_{\Shape{I}}\!\Shape{A}}_{1 + \Pos{I}}(\Pos{A}, \Pos{A}),
  \end{gathered}\end{equation}
  where \((\widetilde{f}{s})(i, b, u) \coloneqq \widetilde{f}(i, {s}(i, b), u)\) and \((\widetilde{r}{f})(i, a, u) \coloneqq \widetilde{r}(i, {f}(i, a), u)\).
  The family of positions of \(A \CEq_{\U^I} B\) is equivalent to the following family over the above characterization of the type of shapes
  \begin{equation}\label{prop:characterization-of-Familial-categorical-equivalence-in-polynomial-model:pos}
    \sum_{\substack{i \co \Shape{I},\, a \co \Shape{A}(i) \\ u \co \Pos{B}(a, {f}(i, a))}} \isInc{0}(\widetilde{f}(i, a, u))
    \quad+
    \sum_{\substack{i \co \Shape{I},\, b \co \Shape{B}(i) \\ u \co \Pos{A}(b, {s}(i, b))}} \isInc{0}(\widetilde{s}(i, b, u))
    \quad+
    \sum_{\substack{i \co \Shape{I},\, b \co \Shape{B}(i) \\ u \co \Pos{A}(b, {r}(i, b))}} \isInc{0}(\widetilde{r}(i, b, u)).
  \end{equation}
\end{lemma}
\begin{proof}
  The type \(\Gamma \vdash A \CEq_{\U^I} B\) is the \(\Sigma\) type given by \(\Gamma \vdash f \co \prod_{i \co I} A(i) \to B(i)\), \(\Gamma \vdash s,r \co \prod_{i \co I} B(i) \to A(i)\), and \(\Gamma \vdash fs = \id\), \(\Gamma \vdash rf = \id\).
  By \Cref{prop:polynomial-model:sigma-types}, the shape component of a \(\Sigma\) type is the \(\Sigma\) type of the shapes, and the position component of a \(\Sigma\) type is given by the coproduct of the positions.
  For the families of functions, these are characterized by \Cref{prop:characterization-of-families-of-functions}.
  By associativity of \(\Sigma\) types, and the curry-uncurry isomorphism these correspond to the six families of functions in~\eqref{prop:characterization-of-Familial-categorical-equivalence-in-polynomial-model:shape-shape} and~\eqref{prop:characterization-of-Familial-categorical-equivalence-in-polynomial-model:shape-pos}.
  
  By \Cref{prop:polynomial-model:identity-types}, the shape of an identity type is the identity type of the shapes.
  As identity types of \(\Sigma\) types, these are equivalent to \(\Sigma\) types of identity types between the first and second component~\cite[Theorem 9.3.4]{Rijke2025}.
  Since identity types respect the strict isomorphism used above up to equivalence, we see that the shape components of the two identity types are equivalent to the four identity types in~\eqref{prop:characterization-of-Familial-categorical-equivalence-in-polynomial-model:shape-shape} and~\eqref{prop:characterization-of-Familial-categorical-equivalence-in-polynomial-model:shape-pos}.

  The composition corresponds to composition in the claimed category by \Cref{prop:characterization-of-shapes-of-composite}.
  By \Cref{prop:polynomial-model:identity-types}, identity types have empty positions yielding together with \Cref{prop:characterization-of-families-of-functions} the above description given in~\eqref{prop:characterization-of-Familial-categorical-equivalence-in-polynomial-model:pos}.
\end{proof}

\begin{remark}\label[remark]{prop:equivalences-in-polynomial-model}
  We sketch the unfolding of \((A \simeq B) \in \Ty[\Poly{\cat{C}}](\Gamma)\).
  The data given by the functions \(f, s, r\) is the same as in \Cref{prop:characterization-of-Familial-categorical-equivalence-in-polynomial-model}.
  The homotopies contribute \(\ShapeShape{{{r}}} \cc \ShapeShape{{{f}}} \sim \id_{\Shape{A}}\) and \(\ShapeShape{{{f}}} \cc \ShapeShape{{{s}}} \sim \id_{\Shape{B}}\) to the shape part.
  Unlike in \Cref{prop:characterization-of-Familial-categorical-equivalence-in-polynomial-model}, however, the homotopies do not encode any relationship between \(\ShapePos{{{f}}}\), \(\ShapePos{{{s}}}\), and \(\ShapePos{{{r}}}\); a homotopy in \(\Poly{\cat{C}}\) unfolds only to a homotopy between shape components in \(\cat{C}\).
  As such, the homotopies in \(\Shape{(A \simeq B)}\) witness an equivalence only between the shape components of \(A\) and \(B\).
  The family of positions of \(A \simeq B\) agrees with that of \(A \CEq B\) in \Cref{prop:characterization-of-Familial-categorical-equivalence-in-polynomial-model}.
\end{remark}

\begin{lemma}[In \(\cat{C}\)]\label[lemma]{prop:good-characterization-of-Familial-categorical-equivalence-in-polynomial-model}
    If \(\cat{C} \vDash \FamCUA_{\U}\), then \(\Shape{\CEq*}(I, A, B) \simeq \paren[\big]{\sum_{e \co \Shape{A} \catIso[\U^{\Shape{I}}] \Shape{B}} \Pos{A} \catIso[\U^{\widetilde{I}}] \Pos{B}e}\) where \(\widetilde{I} \coloneqq \sum_{\Shape{I}}\Shape{A}\).
\end{lemma}
\begin{proof}
  Let \(I \seq (\Shape{I}, \Pos{I}) \co \Shape{\U}\), \(A \seq (\Shape{A}, \Pos{A}), B \seq (\Shape{B}, \Pos{B}) \co \Shape{\Fam}(I)\).
  Set \(J(i) \coloneqq 1 + \Pos{I}(i)\) and \(\widetilde{I} \coloneqq \sum_{i \co \Shape I}\Shape{A}(i)\).
  Note that the components of \(\Shape{\CEq*}(I, A, B)\) given in \Cref{prop:characterization-of-Familial-categorical-equivalence-in-polynomial-model}~\eqref{prop:characterization-of-Familial-categorical-equivalence-in-polynomial-model:shape-shape} are equivalent to \(\Shape{A} \catIso[\U^{\Shape{I}}] \Shape{B}\).
  Denote the remaining components given in \Cref{prop:characterization-of-Familial-categorical-equivalence-in-polynomial-model}~\eqref{prop:characterization-of-Familial-categorical-equivalence-in-polynomial-model:shape-pos} by \(E(I, A, B)\).
  It suffices to give for each \(e \co \Shape{A} \catIso[\U^{\Shape{I}}] \Shape{B}\) an equivalence \(E(I, A, B, e) \simeq \paren{ \Pos{A} \catIso[\U^{\widetilde{I}}] \Pos{B}e }\).
  By the fundamental theorem of identity types~\cite[Theorem 11.2.2]{Rijke2025} and \(\FamCUA_{\U}\), it suffices to consider the case where \(\Shape A \seq \Shape B\) and \(e \seq \id\).
  But in this case \(E(I, A, B, \id)\) reduces to \(\Pos{A} \catIso[\U_J^{\widetilde{I}}] \Pos{B}\) which is equivalent to \(\Pos{A} \catIso[\U^{\widetilde{I}}] \Pos{B}\) by \Cref{prop:type-of-partial-equivalences-is-equivalent-to-equivalences}.
\end{proof}
\begin{lemma}[In \(\cat{C}\)]\label[lemma]{prop:famcua-on-shapes}
  If \(\cat{C} \vDash \FamCUA_{\U}\), then \(\sum_{B \co \Shape{\Fam}(I)} \Shape{\CEq*}(I, A, B)\) is contractible for \(I \co \Shape{\U}\), \(A \co \Shape{\Fam}(I)\).
\end{lemma}
\begin{proof}
  By \Cref{prop:good-characterization-of-Familial-categorical-equivalence-in-polynomial-model} and \(\FamCUA_{\U}\) with \Cref{prop:contractability-characterization-of-univalence}.
\end{proof}

\begin{lemma}\label[lemma]{prop:types-are-prop-iff-shapes-are}
  A type of \(\Poly{\cat{C}}\) is a proposition exactly if its image under \(\Shape{-}\) is: naturally in \(\Gamma \in \Poly{\cat{C}}\), given \(A \in \Ty[\Poly{\cat{C}}](\Gamma)\) there is a logical equivalence
  \(
    \Tm[\cat{C}](\Shape{\Gamma}, \isProp(\Shape A)) \longleftrightarrow \Tm[\Poly{\cat{C}}](\Gamma, \isProp(A)).
  \)
\end{lemma}
\begin{proof}
  By \Cref{prop:polynomial-model:identity-types}, we have that identity types are in the image of \(\Mod\) and preserved by \(\Shape-\).
  Thus, \(\Tm(\Gamma.A.A, \q\p =_A \q) \cong \Tm(\Shape\Gamma.\Shape A.\Shape A, \Shape{\paren{\q\p =_A \q}}) \cong \Tm(\Shape\Gamma.\Shape A.\Shape A, \q\p =_{\Shape A} \q)\).
\end{proof}

\begin{theorem}\label[theorem]{prop:poly-inherits-famcua}
  If \(\cat{C} \vDash \FamCUA_{\U}\), then \(\Poly{\cat{C}} \vDash \FamCUA_{\U}\).
\end{theorem}
\begin{proof}
  By \Cref{prop:contractability-characterization-of-univalence}, it suffices to show \(I \co \U, A \co \Fam(I) \vdash \sum_{B \co \Fam(I)} \CEq*(I,A, B)\) is contractible.
  It is inhabited by the identity, so it is enough to show it is a proposition.
  By \Cref{prop:types-are-prop-iff-shapes-are}, it suffices to show \(I \co \Shape{\U}, A \co \Shape{\Fam(I)} \vdash \sum_{B \co \Shape{\Fam}(I)} \Shape{\CEq*}(I, A, B)\) is a proposition in \(\cat{C}\), and this is \Cref{prop:famcua-on-shapes}.
\end{proof}

\begin{remark}
  Von Glehn \cite[\S5.1]{Glehn2015} observes that the outputs of \(\Poly{-}\) are also suitable inputs to \(\Poly{-}\), meaning the construction can be iterated.
  \Cref{prop:poly-inherits-famcua} implies that iterated polynomial models also inherit \(\FamCUA_{\U}\) from the base model, though we do not know if there is any use for these models.
\end{remark}

\section{Familial categorical univalence without function extensionality}\label{sec:polynomial-refutation}

Using the results from \Cref{sec:categorical-univalence} together with Von Glehn's counterexample to function extensionality in \(\Poly{\cat{C}}\), which we recall in \Cref{sec:failure-of-funext}, we can derive the independence of \(\FE_{\U}\) from \(\ITT+\CCUA_{\U}\).

\subsection{Failure of function extensionality in the polynomial model}\label{sec:failure-of-funext}

Von Glehn's proof that \(\FE\) fails in \(\Poly{\cat{C}}\)~\cite[Proposition 4.11]{Glehn2015} uses the following types:

\begin{definition}\label[definition]{defn:topgun}
  Given \(\Gamma \in \Poly{\cat{C}}\) and \(A \in \Ty[\cat{C}](\Shape{\Gamma})\), define \(\TopGun{A} \in \Ty[\Poly{\cat{C}}](\Gamma)\) by \(\Shape{\Gamma} \vdash \Shape{\TopGun{A}} \coloneqq 1\) and \(\Shape{\Gamma}.\Shape{\TopGun{A}} \vdash \Pos{\TopGun{A}} \coloneqq A\).
\end{definition}

\begin{proposition}\label[proposition]{prop:evil-functions}
  There exist \(f_0,f_1 \in \Tm(1, \TopGun{1+1} \to \TopGun{1})\) with \((f_0 = f_1) \to 0\).
\end{proposition}
\begin{proof}
  For each \(k \in \{0,1\}\), we define \(b_k \in \Tm(1.\TopGun{1+1}, \TopGun{1})\) by setting \(1.1 \vdash \Shape{(b_k)} \coloneqq \star \co 1\) and
  \(1.1.1 \vdash \Pos{(b_k)} \coloneqq \inc{k}(\star) \co 1 + 1\).
  Take \(f_k \coloneqq \lambda(b_k)\).
  Unfolding the construction of \(\lambda\) in \cref{prop:polynomial-model:pi-types}, we have \(\Shape{(f_k)} \seq \pair{\Shape{(b_k)},\Pos{(b_k)}}\), so \(\Shape{(f_0 = f_1)}\) implies \(\Pos{(b_0)} = \Pos{(b_1)}\) and is thus empty.
\end{proof}

\begin{proposition}[{\cite[Proposition 4.11]{Glehn2015}}]\label[proposition]{prop:polynomial-model-no-funext-universe}
  \(\Poly{\cat{C}} \vDash \neg\FE_{\U}\).
\end{proposition}
\begin{proof}
  The functions from \Cref{prop:evil-functions} are homotopic since the codomain is a proposition by \Cref{prop:types-are-prop-iff-shapes-are}.
  Note that \(\TopGun{A}\) belongs to the universe of \(\Poly{\cat{C}}\) for \(A \co \U\).
\end{proof}

\subsection{Independence of function extensionality from familial categorical univalence}

\begin{proposition}\label[proposition]{prop:extensive-univalent-base-model}
  There is a model \(\cat{C}\) of \(\ITT\) with extensive finite coproducts satisfying the strict \(\eta\) rule such that \(\cat{C} \vDash \FE\) and \(\cat{C} \vDash \UA_{\U}\).
\end{proposition}
\begin{proof}
  Take the model of \(\ITT + \FE + \UA_{\U}\) constructed by Cohen, Coquand, Huber, and M\"ortberg \cite{CohenCoquandHuberMortberg2015}, whose category of contexts is the category of presheaves on the De Morgan cube category and whose types are dependent cubical sets equipped with a uniform Kan filling operation.
  Orton and Pitts \cite[Theorem 5.14]{OrtonPitts2018} show that binary coproducts of types can be modeled by coproducts of dependent cubical sets, and it is easy to check the same for nullary coproducts.
  Thus these coproducts satisfy the strict \(\eta\) law and, since every topos is an extensive category \cite[Remark 4.10]{CarboniLackWalters1993}, are also extensive.
\end{proof}

\begin{remark}
  The particular choice of cubical model is not important in the proof above; any model in the style of Orton and Pitts \cite{OrtonPitts2018} or Angiuli et al.\ \cite{AngiuliBrunerieCoquandHarperFavoniaLicata2021} will do, as will Voevodsky's (non-constructive) simplicial model \cite{KapulkinLumsdaine2021}.
  There are, however, models of \(\ITT + \UA_{\U}\) that do not support extensive finite coproducts of types; see for example the need for a factorization in Shulman~\cite[Proposition 6.2]{Shulman2019}.
\end{remark}

\begin{theorem}\label[theorem]{prop:famcua-not-implies-fe}
  \(\ITT + \FamCUA_{\U} \nvdash \FE_{\U}\).
\end{theorem}
\begin{proof}
  Take \(\cat{C}\) to be a model of \(\ITT + \FE + \UA_{\U}\) with extensive finite coproducts of types, as provided by \cref{prop:extensive-univalent-base-model}.
  The combination \(\FE + \UA_{\U}\) implies \(\FamCUA_{\U}\), as \(\FE\) tells us that \((A =_{I \to \U} B) \simeq (A \sim B)\).
  Thus \(\Poly{\cat{C}} \vDash \ITT + \FamCUA_{\U}\) by \Cref{prop:poly-inherits-famcua}, while \(\Poly{\cat{C}} \nvDash \FE_{\U}\) by \cref{prop:polynomial-model-no-funext-universe}.
\end{proof}

\section{Variations}\label{sec:other-univalences}

Once \(\UA_{\U}\) was proposed by Voevodsky, it was quickly taken up as the canonical axiom for its intended purpose.
Inequivalent variations on \(\UA_{\U}\) usually turn out to be significantly weaker, as in case of the ``isomorphism reflection'' that holds in Bauer and Winterhalter's cardinal model \cite[\S8.3]{Winterhalter2020}, or else inconsistent, as in the case of ``\(\mathsf{qinv}\)-univalence'' \cite[Exercise 4.6]{HoTTBook2013}.

Unfortunately, we do not see evidence for a canonical form of ``\(\FE\)-free univalence''.
In this section, we show that a few possible candidates are inequivalent; none stands out as the most natural.
In \Cref{sec:other-univalences:non-familial}, we show that \(\CUA_{\U}\) does not imply \(\FamCUA_{\U}\).
In \Cref{sec:other-univalences:categorical-categorical}, we identify an axiom \(\CCUA_{\U}\) that also satisfies \(\ITT + \CCUA_{\U} \nvdash \FE_{\U}\) and \(\ITT + \CCUA_{\U} + \FE_{\U} \vdash \UA_{\U}\) but is not equivalent to \(\CUA_{\U}\) or \(\FamCUA_{\U}\).

In \Cref{sec:other-univalences:approximate}, we recall a variant of univalence used by Van den Berg \cite[Definition 2.13]{VanDenBerg2020} which we call \emph{approximate univalence} or \(\PUA_{\U}\).
It is an open question whether \(\PUA_{\U}\) implies \(\FE_{\U}\); we do not resolve the question, but we pose a related question that avoids mention of a universe.

\subsection{Non-familial categorical univalence}
\label{sec:other-univalences:non-familial}

Our \Cref{prop:famcua-not-implies-fe} is a priori more than an answer to Dorais' question of whether \(\ITT + \CUA_{\U}\) proves \(\FE_{\U}\): we prove not only that \(\ITT + \CUA_{\U} \nvdash \FE_{\U}\) but that \(\ITT + \FamCUA_{\U} \nvdash \FE_{\U}\).
One may then wonder if \(\FamCUA_{\U}\) is strictly stronger than \(\CUA_{\U}\).
This is indeed the case.

\begin{theorem}\label[theorem]{prop:cua-does-not-imply-famcua}
  \(\ITT + \CUA_{\U} \nvdash \FamCUA_{\U}\).
\end{theorem}
\begin{proof}
  Take \(\cat{C}\) to be a model of \(\ITT + \FE + \UA_{\U}\) with extensive finite coproducts of types, as provided by \cref{prop:extensive-univalent-base-model}.
  Then \(\Poly{\cat{C}} \vDash \FamCUA_{\U}\) by \Cref{prop:poly-inherits-famcua}, and in particular \(\Poly{\cat{C}} \vDash \CUA_{\U}\).
  We now consider the slice model \(\Poly{\cat{C}}/\TopGun{1}\) for \(\TopGun{-}\) from \Cref{defn:topgun}.
  That is, we work in the context \(t \co \TopGun{1}\).
  However, we modify the interpretation of the universe \(\U\).

  Define \((\U',\El')\) by \(\U' \coloneqq \U \times \TopGun{1} \in \Ty[\Poly{\cat{C}}](1)\) with \(\El'\pair{A,t} \coloneqq A \in \Ty[\Poly{\cat{C}}](\U')\).
  In the slice model, \(\TopGun{1}\) is contractible by \Cref{prop:types-are-prop-iff-shapes-are}, so this universe is closed under the same type formers as \(\U\) and the projection \(\pi \co \U' \to \U\) is an equivalence.
  For \(A,B \co \U'\) in the slice model, the interpretation of \(\idToCEq\) for \(\U'\) is homotopic (by path induction) to the composite of \(\ap{\pi} \co (A =_{\U'} B) \to (\pi A =_{\U} \pi B)\) followed by \(\idToCEq \co (\pi A =_{\U} \pi B) \to (\pi A \CEq \pi B)\).
  Since \(\ap{\pi}\) is an equivalence, \(\CUA_{\U'}\) holds in the slice model.

  To see that \(\FamCUA_{\U'}\) fails in the slice model, take \(I \coloneqq \TopGun{1+1}\) and recall from \Cref{prop:evil-functions} that there exist distinct \(f_0 \neq f_1 \co I \to \TopGun{1}\).
  For \(k \in \{0,1\}\), set \(A_k \coloneqq (\lambda i.\pair{1,f_k(i)}) \in (\U')^{I}\).
  Then \(A_0 \cong_{(\U')^{I}} A_1\) is by definition \(1 \cong_{(\U)^I} 1\) and thus inhabited, while \(A_0 =_{(\U')^{I}} A_1\) would imply \(f_0 = f_1\) and is thus empty.
\end{proof}

\subsection{Categorical categorical univalence}\label{sec:other-univalences:categorical-categorical}

In our formulation of \(\CUA_{\U}\), we could have required that \(\idToCEq\) be a \emph{categorical} equivalence.

\begin{definition}\label{defn:categorical-categorical-univalence}
  \emph{Categorical categorical univalence} (\(\CCUA_{\U}\)) is the principle that the canonical map \(\idToCEq \co (A =_{\U} B) \to (A \catIso B)\) is a categorical equivalence for all \(A,B \co \U\).
\end{definition}

A point in favor of \(\CCUA_{\U}\) is that it is a proposition (cf.~\Cref{prop:being-isomorphism-is-proposition}); the structure of ``being an equivalence'' need not be a proposition without \(\FE\) (cf.~implication \eqref{prop:funext-characterizations:isequiv-is-prop} \(\implies\) \eqref{prop:funext-characterizations:funext} of \Cref{prop:funext-characterizations}).
However, it is unusually strong relative to other identity type characterizations.
For example, the equivalence \((\pair{a,b} =_{A \times B} \pair{a',b'}) \simeq (a =_{A} a') \times (b =_{B} b')\) characterizing identities in \(\Sigma\) types cannot be shown to be a categorical equivalence in \(\ITT\).
\(\CCUA_{\U}\) also seems brittle.
Note that the ``canonical'' map \(\idToCEq \co A =_{\U} B \to A \catIso B\) is only canonically defined \emph{up to homotopy} by the requirement \(\idToCEq(\refl[A]) = \id\)!
It is not clear to us that different formulations of \(\CCUA_{\U}\) using homotopic definitions of \(\idToCEq\) are interderivable.

In any case, using \(\Poly{-}\), we will show that \(\CCUA_{\U}\) is strictly stronger than \(\CUA_{\U}\) and moreover not implied by \(\FamCUA_{\U}\), yet still does not imply \(\FE_{\U}\).
We strengthen \Cref{prop:poly-inherits-famcua} to \(\CCUA_{\U}\) by exploiting properties of types in the essential image of \(\Mod\).
By general properties of reflective subcategories, these are exactly those with strictly invertible unit.
In fact, they are also exactly those with categorically invertible unit, but we will not need this.

\begin{proposition}\label[proposition]{prop:modal-types-characterization}
  Naturally in \(\Gamma\), for \(A \in \Ty[\Poly{\cat{C}}](\Gamma)\), the following are equivalent:
  \begin{enumerate}
    \item\label{prop:modal-types-characterization:strict-isomorphism} \(\eta_A \co A \to \Mod A\) is a strict isomorphism,
    \item\label{prop:modal-types-characterization:sub-initial-positions} there is a map \(\Pos{A} \to 0\) in \(\Ty*[\cat{C}](\Shape\Gamma.\Shape{A})\),
    \item\label{prop:modal-types-characterization:initial-positions} \(\Pos{A}\) is strictly isomorphic to \(0\) in \(\Ty*[\cat{C}](\Shape\Gamma.\Shape{A})\).
  \end{enumerate}
\end{proposition}
\begin{proof}
  Consider the morphism \(\eta_A \co \Gamma.A \to \Gamma.\Mod A\) in \(\Poly{\cat{C}}\) over \(\Gamma\).
  The shape component is given by \(\id_A \co \Shape\Gamma.\Shape A \to \Shape\Gamma.\Shape A\).
  The positions component is given by \([\inc{0}, !_{\Pos{A}}] \co \Shape\Gamma.\Shape{A}.\Pos\Gamma\p + 0 \to \Pos\Gamma.\Shape{A}.\Pos\Gamma\p + \Pos{A}\).
  Since the shape component is an isomorphism, \(\eta_A\) is an isomorphism exactly if the position component is.

  We work in the internal language of \(\cat{C}\).
  The direction~\eqref{prop:modal-types-characterization:initial-positions}~\(\implies\)~\eqref{prop:modal-types-characterization:strict-isomorphism} is clear.
  The equivalence~\eqref{prop:modal-types-characterization:sub-initial-positions}~\(\iff\)~\eqref{prop:modal-types-characterization:initial-positions} follows from the strict \(\eta\) rule for \(0\).
  It is left to show~\eqref{prop:modal-types-characterization:strict-isomorphism}~\(\implies\)~\eqref{prop:modal-types-characterization:sub-initial-positions}.
  Suppose we are given a family of strict inverses \(\lambda a. [\inc{0}, i_a] \co \prod_{a \co \Shape{A}} \Pos\Gamma + \Pos{A}(a) \to \Pos\Gamma + 0\) to \(\lambda a. [\inc{0}, \elim_0] \co \prod_{a \co \Shape{A}} \Pos\Gamma + 0 \to \Pos\Gamma + \Pos{A}(a)\).
  Then \(\lambda a.i_a \co \prod_{a \co \Shape{A}} \Pos{A}(a) \to \Pos{\Gamma} + 0\) and \(\lambda a.\elim_0 \co \prod_{a \co \Shape{A}} 0 \to \Pos{\Gamma} + \Pos{A}(a)\) form an equivalence in \(\U^{\Shape{A}}_{\Pos{\Gamma}}\).
  Hence, the family of maps \(i\) is total by \Cref{prop:partial-equivalences-are-total}.
\end{proof}

\begin{definition}
  A type \(A \in \Ty[\Poly{\cat{C}}](\Gamma)\) is \emph{\(\Mod\)-modal} if the conditions from \Cref{prop:modal-types-characterization} hold.
\end{definition}

Let \(A, B \in \Ty[\cat{C}](\Gamma)\), \(f \in \Tm[\cat{C}](\Gamma, A \to B)\), and \(F \co \cat{C} \to \cat{D}\) a pseudomorphism.
We write
\[
  \widetilde{F} \co \Tm(\Gamma, A \to B) \longrightarrow \Tm(F\Gamma, FA \to FB), \qquad f \longmapsto \lambda(F(\app(f))).
\]
The image of \(f\) under \(F \co \cat{C} \to \cat{D}\) is \(Ff \in \Tm(F\Gamma, F(A \to B))\).
There is always a comparison map \(\lambda(F (\app(\q_{A\to B}))) \co {F(B^A)} \to (F B)^{(F A)}\)~\cite[\S4]{KaposiHuberSattler2019}, and the image of \(Ff\) under this map coincides with \(\widetilde{F}f\).

\begin{lemma}
  The pseudomorphism \(\Mod \co \cat{C} \to \Poly{\cat{C}}\) preserves paths between functions:
  naturally in \(\Gamma\), given \(A, B \in \Ty(\Gamma)\) and \(f, g \in \Tm(\Gamma, A \to B)\), there is a map 
  \[
    \Tm(\Gamma, f =_{A \to B} g) \longrightarrow \Tm(\Mod\Gamma, \ModF {f} =_{\Mod A \to \Mod B} \ModF g).
  \]
\end{lemma}
\begin{proof}
  Let \(H \in \Tm(\Gamma, f = g)\).
  We have \(\Mod H \in \Tm(\Mod\Gamma, \Mod\paren{f =_{A \to B} g})\).
  By the definition of the identity types in \(\Poly{\cat{C}}\) (\Cref{prop:polynomial-model:identity-types}), they are preserved by the pseudomorphism \(\Mod\).
  In particular, we have \(\Tm(\Mod\Gamma, \Mod\paren{f =_{A \to B} g}) \cong \Tm(\Mod\Gamma, \Mod f =_{\Mod(A \to B)} \Mod g)\).
  By lifting the comparison map \({\Mod(B^A)} \to (\Mod B)^{(\Mod A)}\) to identity types, we obtain the desired element.
\end{proof}

\begin{corollary}\label[corollary]{prop:inclusion-preserves-categorical-equivalences}
  The pseudomorphism \(\Mod \co \cat{C} \to \Poly{\cat{C}}\) preserves categorical equivalences:
  naturally in \(\Gamma\), given \(A, B \in \Ty(\Gamma)\), \(f \in \Tm(\Gamma, A \to B)\) we have a map 
  \[
    \Tm(\Gamma, \isCEq(f)) \longrightarrow \Tm(\Mod \Gamma, \isCEq(\ModF f)).
  \]  
\end{corollary}
\begin{proof}
  The action on functions \(\ModF\) preserves composition and identities.
\end{proof}

The \(\Mod\)-modal types in \(\Poly{\cat{C}}\) behave like types of the base model.
In particular, when the base model enjoys function extensionality, homotopy equivalences between \(\Mod\)-modal types can be improved to categorical equivalences.

\begin{lemma}\label[lemma]{prop:on-fe-model-modal-types-rectify-equivalences}
  If \(\cat{C} \vDash \FE\), then equivalences coincide with categorical equivalences between \(\Mod\)-modal types.
\end{lemma}
\begin{proof}
  Let \(A, B \in \Ty[\Poly{\cat{C}}](\Gamma)\) and \(f \co A \to B\) an equivalence.
  Since \(\Shape{-}\) preserves identity types, the map \(\Shape{f} \co \Shape{A} \to \Shape{B}\) is an equivalence, and by \(\FE\) also a categorical equivalence.
  The pseudomorphism \(\Mod\) preserves categorical equivalences by \Cref{prop:inclusion-preserves-categorical-equivalences}.
  Hence, so does the dependent right adjoint \(\Mod\) since it is defined by substituting the image of the pseudomorphism along the unit.
  Since \(\Mod f \cc \eta_A \seq \eta_B \cc f\) the claim follows by \Cref{prop:modal-types-characterization} and \(2\)-out-of-\(3\) for categorical equivalences (\Cref{prop:isomorphisms-satisfy-2-out-of-3}).
\end{proof}

We show that the type family \(I \co \U, A, B \co \Fam(I) \vdash \CEq*(I, A, B) \coloneqq (A \CEq_{\U^I} B)\) from \cref{sec:familial-categorical-univalence} is \(\Mod\)-modal, which we can use to improve the equivalence in the definition of \(\FamCUA_{\U}\) to a categorical equivalence over well-behaved base models.

\begin{lemma}\label[lemma]{prop:type-of-familial-categorical-equivalences-is-modal}
  If \(\cat{C} \vDash \FamCUA_{\U}\), then \(\CEq*\) in \(\Poly{\cat{C}}\) is \(\Mod\)-modal.
\end{lemma}
\begin{proof}
  By \Cref{prop:modal-types-characterization}, it suffices to show that the positions of \(\CEq* \in \Ty(1.\U.\Fam \times \Fam, \CEq*)\) are empty.
  We work internally to \(\cat{C}\).
  We show for all \(I \co \Shape\U\), \(A \seq (\Shape A, \Pos A), B \seq (\Shape B, \Pos B) \co \Shape{\Fam}(I)\), \(e \co \Shape{\CEq*}(I, A, B)\) that \(\Pos{\CEq*}(I, A, B, e) \to 0\), which suffices by strict initiality of \(0\).
  By \Cref{prop:famcua-on-shapes}, we can assume \(A \seq B\) and \(e \seq \id\).
  In this case, the functions on positions are given by \(\inc{1}\) and therefore the type \eqref{prop:characterization-of-Familial-categorical-equivalence-in-polynomial-model:pos} is empty.
\end{proof}

\begin{proposition}
  \label[proposition]{prop:poly+fe-inherits-ccua}
  If \(\cat{C}\) is a model of \(\ITT + \FE + \UA_{\U}\) with extensive finite coproducts of types, then \(\Poly{\cat{C}} \vDash \CCUA_{\U}\).
\end{proposition}
\begin{proof}
  By \Cref{prop:poly-inherits-famcua}, we have \(\Poly{\cat{C}} \vDash \CUA_{\U}\).
  Since the types declared equivalent in the statement of \(\CUA_{\U}\) are \(\Mod\)-modal, by \Cref{prop:type-of-familial-categorical-equivalences-is-modal} and the definition of identity types (\Cref{prop:polynomial-model:identity-types}), the claim follows from \Cref{prop:on-fe-model-modal-types-rectify-equivalences}.
\end{proof}

\begin{corollary}
  \(\ITT + \CCUA_{\U} \nvdash \FE_{\U}\).
\end{corollary}
\begin{proof}
  \Cref{prop:extensive-univalent-base-model} provides a model \(\cat{C}\) of of \(\ITT + \FE + \UA_{\U}\) with extensive finite coproducts of types.
  We have \(\Poly{\cat{C}} \vDash \CCUA_{\U}\) by \Cref{prop:poly+fe-inherits-ccua} and \(\Poly{\cat{C}} \vDash \lnot\FE_{\U}\) by \Cref{prop:polynomial-model-no-funext-universe}.
\end{proof}

Finally, we show that while \(\CCUA_{\U}\) is still strictly weaker than \(\UA_{\U}\), it is strictly stronger than \(\CUA_{\U}\).

\begin{proposition}\label[proposition]{prop:categorical-univalence-does-not-imply-familial}
  \(\ITT + \FamCUA_{\U} \nvdash \CCUA_{\U}\).
\end{proposition}
\begin{proof}
  Take \(\cat{C}\) to be a model of \(\ITT + \FE + \UA_{\U}\) with extensive finite coproducts of types, as provided by \cref{prop:extensive-univalent-base-model}.
  Then \(\Poly{\cat{C}} \vDash \FamCUA_{\U}\) by \Cref{prop:poly-inherits-famcua}.
  As in the proof of \cref{prop:cua-does-not-imply-famcua}, we now consider the slice model \(\Poly{\cat{C}}/\TopGun{1}\) for \(\TopGun{-}\) from \Cref{defn:topgun}.
  This time, however, we modify the interpretation of the identity type.

  We define our new identity types by \((u ='_A v) \coloneqq (u =_A v) \times \TopGun{1}\), where \(u =_A v\) is the identity type in \(\Poly{\cat{C}}\).
  Since \(\TopGun{1}\) is a proposition by \Cref{prop:types-are-prop-iff-shapes-are}, it is contractible in the slice, so the projection \((u ='_A v) \to (u =_A v)\) is an equivalence and in particular \(u ='_A v\) is an identity type.
  If we write \(\cong'\) for wild-categorical isomorphisms defined with \(='\), it follows also that \((a \cong_{{\cat{D}}} b) \simeq (a \cong'_{{\cat{D}}} b)\) for any wild category \(\cat{D}\).
  Thus \(\FamCUA_{\U}\), which holds in \(\Poly{\cat{C}}\) by \Cref{prop:poly-inherits-famcua}, transfers to the slice model with the new identity type.

  However, \(\CCUA_{\U}\) cannot hold (in or out of the slice) when formulated with \(='\) and \(\cong'\).
  For \(A,B \co \U\), the family of positions for \((A ='_{\U} B)\) is the constant family \(1\).
  The family of positions for \(A \cong'_{\U} B\), which is categorically equivalent to \((A \cong_{\U} B) \times \TopGun{1} \times \TopGun{1}\), is the constant family \(1 + 1\) (using \Cref{prop:type-of-familial-categorical-equivalences-is-modal}).
  Thus, by \Cref{prop:good-characterization-of-Familial-categorical-equivalence-in-polynomial-model}, the equivalence \((A ='_{\U} B) \simeq (A \catIso' B)\) is only categorical when both sides are empty.
\end{proof}

\subsection{Approximate univalence}
\label{sec:other-univalences:approximate}

Van den Berg \cite[Definition 2.13]{VanDenBerg2020} defines another weak form of \(\UA_{\U}\), in the language of path categories, which can be rendered in type theory as follows.

\begin{definition}\label{defn:approximate-univalence}
  \emph{Approximate univalence} (\(\PUA_\U\)) is the principle that for all \(A,B \co \U\) and \(e \co A \simeq B\), we have some \(p \co A =_{\U} B\) such that \(\idToEquiv(p) \sim e\).
\end{definition}

Notably, \(\PUA_\U\) can be expressed as an inference rule without \(\Pi\) types.
In the presence of \(\Pi\) types, Swan \cite[Remark 4.6]{Swan2024} comments that it is an open question whether \(\PUA_{\U}\) implies \(\FE_{\U}\).
An immediate but subtle consequence of \(\PUA_{\U}\) is that there is a composite map \((A \simeq B) \to (A =_{\U} B) \to (A \cong B)\) that improves any homotopy equivalence to a homotopic categorical equivalence.
In light of the decomposition of \(\UA_\U\) in \Cref{sec:decomposing-univalence:funext}, it is natural to consider an analogue of \Cref{defn:equivalence-improvement}:

\begin{definition}\label{defn:approximate-equivalence-improvement}
  \emph{Approximate equivalence improvement} (\(\PEI\)) is the principle that for all types \(A,B\) and \(e \co A \simeq B\), we have some \(e' \co A \CEq B\) such that \(\CEqToEq(e') \sim e\).
\end{definition}

One \(\FE\)-like corollary of \(\PEI\) is that if \(P\) is a contractible type, then \(A \to P\) is also contractible for every type \(A\): by \Cref{prop:cat-isomorphism-characterization}, we have \((A \to P) \simeq (A \to 1) \sequiv 1\).
This is not provable in \(\ITT\), as \Cref{prop:types-are-prop-iff-shapes-are,prop:evil-functions} show.
However, the exact relationship between \(\PEI\) and \(\FE\) is a mystery to us:

\begin{ques}\label[question]{prop:pei-implies-fe}
  Does \(\ITT + \PEI \vdash \FE\)?
\end{ques}

An answer to \Cref{prop:pei-implies-fe} might not tell us whether \(\ITT + \PUA_{\U} \vdash \FE_{\U}\), but it may be a more tractable question.
The polynomial models refute \(\PEI\): \(\TopGun{1}\) and \(\TopGun{1+1}\) are equivalent (\Cref{prop:equivalences-in-polynomial-model}) but not categorically equivalent (\Cref{prop:good-characterization-of-Familial-categorical-equivalence-in-polynomial-model}).
Boulier, P\'{e}drot, and Tabareau's \emph{intensional function translation} \cite[\S3]{BoulierPedrotTabareau2017} sends a theory with \(\FE\) to a syntactic model with \(\PEI \land \lnot\FE\), but its function types do not satisfy any \(\eta\) rule, so this does not answer the question for \(\ITT\) as we define it.
Shulman \cite{Shulman2014} has a recipe for expressing universal properties without \(\FE\) that suggests stronger forms of \(\PEI\); for example, one can also ask that homotopic categorical equivalences are equal.
It is not clear to us how these strengthenings relate to \(\PEI\) or to \(\FE\).

\begin{remark}\label{defn:approximate-categorical-univalence}
  Naturally, we can also consider \emph{approximate categorical univalence} \(\PCUA_{\U}\): the principle that for all \(A,B \co \U\) and \(e \co  A \CEq B\), we have some \(p \co A =_{\U} B\) such that \(\idToCEq(p) \sim e\).
  This is the weakest of all the univalence principles we have considered, but we do not know if it is strictly weaker than \(\CUA_{\U}\).
\end{remark}

\section{Related work}\label{sec:related-work}

To conclude, we comment on the status of weak forms of univalence in other known models of type theory without function extensionality.

\subsection{Realizability models}

Realizability is a standard source of models of \(\ITT\) that refute extensionality principles, including \(\FE\); see Streicher \cite[Theorem 2.9, \S3.7]{Streicher1993}.
However, most work combining features of realizability and homotopical semantics, such as that of Frumin and Van den Berg \cite{FruminVanDenBerg2018} and Uemura \cite{Uemura2019}, constructs models that \emph{do} satisfy \(\FE\).
An exception is Speight's \emph{groupoidal realizability} \cite{Speight2024}; his function types have neither \(\FE\) nor the \(\eta\) rule.
Speight constructs an impredicative universe of modest fibrations, but we do not know if this or any other universe in the model satisfies some kind of univalence.

\subsection{P{\'{e}}drot and Tabareau's parametric exceptional translation}

The \emph{parametric exceptional translation} \cite{PedrotTabareau2018} is another source of models of type theory without \(\FE\).
Presented as a syntactic translation, it induces a construction \(\ParEx{-}\) on models.
Unlike \(\Poly{-}\), however, \(\ParEx{-}\) does not preserve any form of univalence that we know of.
We sketch here a reason for the simplest form of the translation (\(\mathbb{E} = 1\) and \(\Omega_i(\star) = 1\)).
Kov{\'a}cs \cite{Kovacs2024antifunext} has formalized this case in Agda.

Given a model \(\cat{C}\), the category of contexts in \(\ParEx{\cat{C}}\) is \(\int_{\Gamma \in \cat{C}} \Ty*(\Gamma)\): objects are pairs \(\Gamma = (\Shape{\Gamma},\Pos{\Gamma})\) as in \(\Poly{\cat{C}}\), but a morphism \(\sigma \co \Delta \to \Gamma\) is a pair of \(\Shape{\sigma} \co \Shape{\Delta} \to \Shape{\Gamma}\) in \(\mathbb{C}\) and \(\Pos{\sigma} \co \Pos{\Delta} \to \Pos{\Gamma}\Shape{\sigma}\) in \(\Ty*(\Shape{\Delta})\).
We think of \(\Pos{\Gamma}\) as selecting ``valid'' elements of \(\Shape{\Gamma}\).
Types \(A \in \Ty(\Gamma)\) have components \(\Shape{A} \in \Ty(\Shape{\Gamma})\), \(\Pos{A} \in \Ty(\Shape{\Gamma}.\Pos{\Gamma}.\Shape{A})\), and \(\Err A \in \Tm(\Gamma, \Shape{A})\).
Again we think of \(\Pos{A}\) as selecting valid elements of \(\Shape{A}\), while \(\Err{A}\) is a distinguished ``error'' element.
While intuitively \(\Err{A}\) should not be valid, this is not enforced.

The mismatch between \(A \CEq B\) and \(A =_\U B\) is that for the former, categorical equivalences of the \(\Shape{-}\) and \(\Pos{-}\) components suffice, while the latter requires also that \(\Err{A}\) corresponds to \(\Err{B}\).
For example, define \(X^0,X^1 \in \Ty(1)\) by \(\Shape{X^k} \coloneqq 1 + 1\), \(\Pos{X^k}(b) \coloneqq 1\), and \(\Err{X^k} = \inc{k}(\star)\).
The identity equivalence on \(1 + 1\) defines a strict isomorphism \(X^0 \sequiv X^1\) that cannot induce a path \(X^0 =_\U X^1\) because it does not send \(\Err{X^0}\) to \(\Err{X^1}\).

\subsection{Bordg's projective model}

Bordg~\cite{Bordg2015,Bordg2017} describes a model of type theory with \(\Sigma\) types, \(\Pi\) types, identity types, and a universe \(\U\) in the category \([\B{C_2},\Gpd]\) of groupoid-valued presheaves on the two-element group.
This model is based on the \emph{projective Quillen model structure}: types are morphisms of \([\B{C_2},\Gpd]\) whose underlying \(\Gpd\)-morphisms are isofibrations.
Bordg observes that the model refutes \(\FE\) \cite[Proposition 6.3]{Bordg2017} and that the natural choice of universe is not univalent \cite[Proposition 6.1]{Bordg2017}, despite the fact that a stronger condition on types corresponding to the \emph{injective model structure} yields a model of \(\FE\) and \(\UA_{\U}\) \cite[\S5.4]{Bordg2015}.

Considering our weaker forms of \(\UA_{\U}\) in this model is fruitless: although \(\FE\) fails, \(\FE_{\U}\) holds (cf.\ \cite[Remark 6.4]{Bordg2017}).
This is because the groupoids classified by \(\U\) are strict sets, as in Hofmann and Streicher's groupoid model \cite{HofmannStreicher95}.
Thus, \(\CUA_{\U}\) for example cannot hold, for if it did then \(\UA_{\U}\) would follow.

\bibliographystyle{./entics}
\bibliography{references}

\end{document}